\documentclass[10pt]{amsart}
\usepackage{amssymb}
\usepackage{paralist} 
\usepackage{stmaryrd}
 \usepackage{graphicx}
 \usepackage{tikz, xcolor}
\usetikzlibrary {shapes}
 \usepackage{array}
 \usepackage{fge}
 \usepackage{multicol}
 \usepackage{caption,pgf} 
\definecolor{lightgray}{gray}{0.9}
 
\definecolor{dnrbl}{rgb}{0,0,0.5}
\definecolor{dnrgr}{rgb}{0,0.5,0}
\definecolor{dnrre}{rgb}{0.5,0,0}
\usepackage[colorlinks=true, citecolor=dnrgr, linkcolor=dnrre, urlcolor=dnrbl]{hyperref}
\usepackage{parskip}
\usepackage{pdfsync}

 \usepackage{wrapfig}

\theoremstyle{plain}
\newtheorem{thm}{Theorem}[section]
\newtheorem{prop}[thm]{Proposition}
\newtheorem{obs}[thm]{Obervation}
\newtheorem{lem}[thm]{Lemma}

\newtheorem{defin}[thm]{Definition}
\numberwithin{equation}{thm}

\newcommand{\bpm}{\begin{pmatrix}}
\newcommand{\epm}{\end{pmatrix}}

\begin{document}

\title[Unperturbed Schelling Segregation in 2D and 3D]{Unperturbed Schelling segregation in two or three dimensions}

\author{George Barmpalias}
\address{{\bf George Barmpalias:} 
(1) State Key Lab of Computer Science, 
Institute of Software,
Chinese Academy of Sciences,
Beijing 100190,
China and (2) School of Mathematics, Statistics and Operations Research,
Victoria University, Wellington, New Zealand}
\email{barmpalias@gmail.com}
\urladdr{\href{http://barmpalias.net}{http://barmpalias.net}}

\author{Richard Elwes}
\address{{\bf Richard Elwes}, School of Mathematics,
University of Leeds, LS2 9JT Leeds, U.K.}
\email{r.elwes@gmail.com}
\urladdr{\href{http://richardelwes.co.uk}{http://richardelwes.co.uk}}

\author{Andy Lewis-Pye}
\address{{\bf Andy Lewis-Pye},  Mathematics, Columbia House, 
London School of Economics, WC2A 2AE, London, U.K.}
 \email{andy@aemlewis.co.uk}
 \urladdr{\href{http://aemlewis.co.uk/}{http://aemlewis.co.uk}}

 \date{\today}
\thanks{Authors are listed alphabetically. 
Barmpalias was supported by the 
1000 Talents Program for Young Scholars from the Chinese Government,
and the Chinese Academy of Sciences (CAS) President's International 
Fellowship Initiative No. 2010Y2GB03.
Additional support was received by
the CAS and the Institute of Software of the CAS.
Partial support was also received from a Marsden grant of New Zealand 
and the China Basic Research Program (973) grant No.~2014CB340302. 
Andy Lewis-Pye (previously  
Andrew Lewis) was supported by a Royal Society University Research Fellowship.
The authors thank Sandro Azaele for some helpful comments on an earlier draft of this paper.}

\begin{abstract}
Schelling's models of segregation, first described in 1969 \cite{TS1} are among the best known models of self-organising behaviour. Their original purpose was to identify mechanisms of urban racial segregation. But his models form part of a family which arises in statistical mechanics, neural networks, social science, and beyond, where populations of agents interact on networks.  Despite extensive study, \emph{unperturbed} Schelling models have largely resisted rigorous analysis, prior results generally focusing on variants in which noise is introduced into the dynamics, the resulting system being amenable to standard techniques from statistical mechanics or stochastic evolutionary game theory \cite{HY}. A series of recent papers \cite{BK,BEL1,BEL2}, has seen the first rigorous analyses of 1-dimensional unperturbed Schelling models, in an asymptotic framework largely unknown in statistical mechanics. Here we provide the first such analysis of 2- and 3-dimensional unperturbed models, establishing most of the phase diagram, and answering a challenge from \cite{BK}.  
\end{abstract}
\maketitle
\setcounter{tocdepth}{1}

\section{Introduction} 
Schelling's spatial proximity models of segregation \cite{TS1} provided strikingly simple examples the emergence of order from randomness. Subsequently, and particularly following Young \cite{HY}, a range of variations of his models have been studied by mathematicians, statistical physicists, computer scientists, social scientists, and others, making Schelling segregation one of the best known theoretical examples of self-organising behaviour (indeed, this was cited by the committee upon awarding Schelling the Nobel memorial prize for Economics in 2005).

Schelling models have the following general set-up: a population of individuals of two types are initially randomly distributed on a grid (or other graph). Each individual considers a certain region around it to be its \emph{neighbourhood}, and it has an \emph{intolerance} ($\tau$), expressing the proportion of its neighbours it requires to be of its own type in order to be happy.

Unhappy individuals then rearrange themselves in order to become happy. This may happen through a variety of possible dynamics. In Schelling's own work on 2-dimensional segregation, at each time step a selected unhappy node would move to a vacant position where it would be happy. However, much subsequent work on Schelling segregation has followed the breakthrough analysis of Young \cite{HY} who studied a 1-dimensional model which dispensed with vacancies, where agents instead rearranged themselves via pairwise swaps (at each time-step a randomly selected pair of unhappy agents of opposite types swap positions).

The lack of vacancies is arguably a closer approximation to life in modern cities without large numbers of uninhabited properties, and has proved itself more amenable to mathetatical analysis. Numerous variants of Young's model have subsequently been investigated, notably in the work of Zhang \cite{JZ1,JZ2,JZ3}, as well as the series of papers of which the current work is a part \cite{BK,BEL1,BEL2}.

Young's pairwise-swapping dynamic is known to statistical physicists as the \emph{Kawasaki dynamic}. An even simpler choice, adopted in \cite{BEL2} and the current paper, is the \emph{Glauber dynamic}, whereby at each time-step a single randomly selected unhappy agent switches type. Although simple, this is by no means unintuitive, the idea is that an agent who becomes unhappy moves out of the city, and is replaced by an agent of the opposite type. Thus our model is an open system, and we posit limitless pools of both types beyond the system. The Glauber dynamic also, we argue, brings the model closest to others in the same statistical-physical family, such as the Ising and Hopfield models (see \ref{subsection:science} below).

Ever since Schelling's original simulations with zinc and copper coins on a checkerboard back in 1969, the robustness of the phenomenon of \emph{segregation} has been recognised: as the process unfolds, large clusters emerge, consisting of individuals of only one type. This same phenomenon has appeared in diverse variants of the model, strikingly including models in which agents have an active preference for integration \cite{JZ2,PV}. Seen from a game theoretic perspective, this provides an example of a recurrent theme in Schelling's research -- especially as elaborated upon in \cite{TS2} --  that individuals acting according to their interests at the local level can produce global results which may be unexpected and  undesired by all. Understanding this phenomenon has been the goal of much subsequent research. 

Besides changing the dynamic, Young in \cite{HY}, made another highly influential innovation by introducing noise to the model. Much subsequent research has focussed on such \emph{perturbed} variants, which can be thought of as systems of temperature $T>0$, in which agents have a small but non-zero probability of acting against their own interests. Young used techniques from evolutionary game theory -- an analysis in terms of stochastically stable states -- to analyse such a model. These ideas were then substantially developed in the work of Zhang \cite{JZ1,JZ2,JZ3}. While the language used may differ from that of those of statistical physicists, the basic analysis is essentially equivalent: Zhang establishes a Boltzmann distribution for the set of configurations, and then his stochastically stable states correspond to ground states.

\subsection{{\bf Our contribution}}

The difficulty in analysing the \emph{unperturbed} (or temperature $T=0$) variants of the model stems from the large number of absorbing states for the underlying Markov process. Nevertheless, a breakthrough came in \cite{BK}, where Brandt, Immorlica, Kamath and Kleinberg used an  analysis of  locally defined stable configurations, combined with results of Wormald \cite{NW}, to provide the first rigorous analysis of an unperturbed 1-dimensional Schelling model, for the case $\tau=0.5$ under Kawasaki dynamics. In \cite{BEL1} the authors gave a more general analysis of the same model  for $\tau\in [0,1]$. In \cite{BEL2}, the authors then analysed variants of the model under Glauber dynamics (as well as variations thereupon), with an additional innovation: the two types of agent may have unequal intolerances $\tau_\alpha, \tau_\beta \in [0,1]$.  We extend the analysis of this phenomenon in the current paper. This modification of the model has justification in social research. See for example \cite{SSBK}, where it is found that black US citizens are happier in integrated neighbourhoods than their white compatriots. It has often been argued (including by the authors in \cite{BEL2} and Schelling himself in for example \cite{TS2}) that the simplicity of Schelling models make them relevant in areas far beyond the topic of racial segregation. Another application might be that of product adoption: which of two competing products customers choose. The peer-effect aspect of this phenomenon can be addressed with a Schelling model. In such a context, the asymmetry of $\tau_\alpha, \tau_\beta$ reflects the idea of `difficulty of engagement', which is argued in \cite{By} to offer insights into the obstacles that innovative ecological products face in terms of market penetration.

Throughout this series of papers, although the proofs vary widely in their details, the analysis hinges on the use of a spread of models across which the neighbourhood radius ($w$) size grows large (but remains small relative to the size of the whole system, $n$). This allows exact results to be obtained asymptotically, by careful probabilistic analysis of various structures occurring in the initial configuration. (It might be objected, at this point, that neighbourhoods of arbitrarily large size are not in the spirit of models of racial segregation. However, this should be seen as a mathematical device allowing the outcome to be predicted with arbitrarily high precision. In simulations of all models in this family, relatively small values of $w$ are generally sufficient for a clear outcome. See Figures \ref{fig:evolvingproc}, \ref{fig:evolvingproc2}, \ref{fig:evolvscenario} for examples.)

The current paper uses the same approach to extend these results to the two and three dimensional models. A mathematical account of unperturbed Schelling segregation in higher dimensions has been seen as a challenging problem for some time. As Brandt, Immorlica, Kamath and Kleinberg say in \cite{BK}:
 
\begin{quote}
{\footnotesize Finally, and most ambitiously, there is the open problem of rigorously analyzing the Schelling model in other graph structures including two-dimensional grids. Simulations of the Schelling model in two dimensions reveal beautiful and intricate patterns that are not well understood analytically. Perturbations of the model have been successfully analyzed using stochastic stability analysis [\dots] but the non-perturbed model has not been rigorously analyzed. Two-dimensional lattice models are almost always much more challenging than one-dimensional ones, and we suspect that to be the case with Schelling's segregation model. But it is a challenge worth undertaking: if one is to use the Schelling model to gain insight into the phenomenon of residential segregation, it is vital to understand its behavior on two-dimensional grids since they reflect the structure of so many residential neighborhoods in reality.}
\end{quote}   
 \begin{figure}
 \centering
 \includegraphics[width=1.3in]{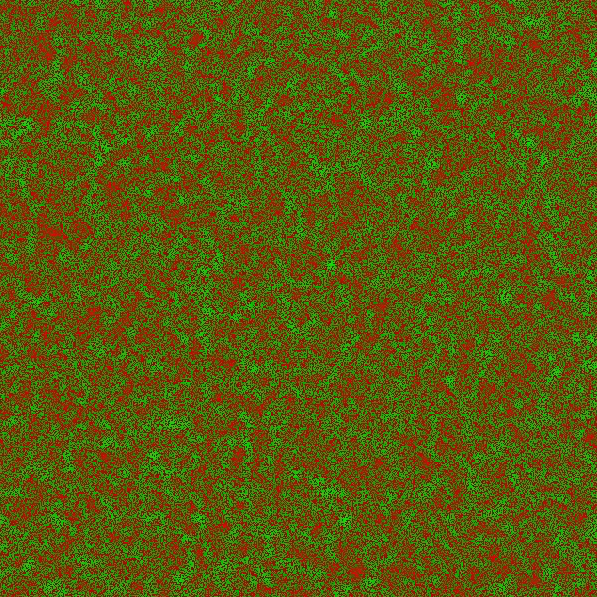}
 \includegraphics[width=1.3in]{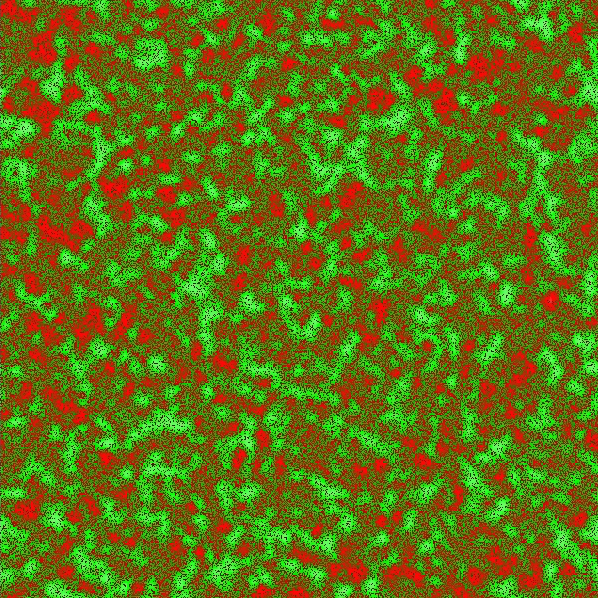}
  \includegraphics[width=1.3in]{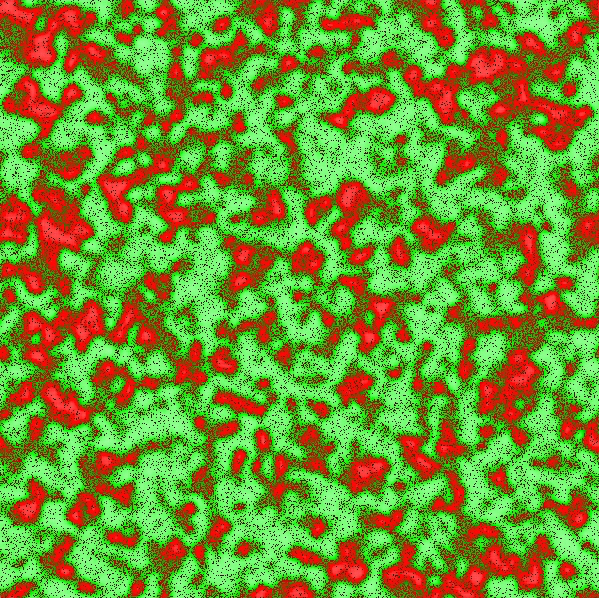}
   \includegraphics[width=1.3in]{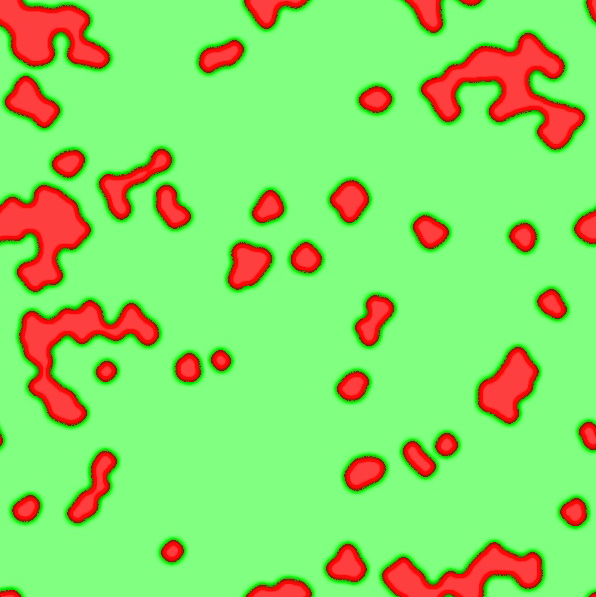}
\caption{The evolving process: $n=600,w=5,\tau_{\alpha}=0.44,\tau_{\beta}=0.42$.}
\label{fig:evolvingproc}
\end{figure}

Thus the aim of this paper is to answer this challenge, providing the first rigorous analysis of the two and three
dimensional unperturbed models of Schelling segregation. Our work  gives an almost comprehensive picture of the phase
diagram of such a model under Glauber dynamics with distinct intolerances $\tau_\alpha$ and $\tau_\beta$, revealing interesting phase transitions around certain thresholds, which are solutions to derivable equations and which we numerically approximate. Some grey areas around these thresholds remain, where the behaviour of the model is not proven rigorously, but these correspond to relatively small intervals.

Let us briefly compare the results of the current paper (namely Theorems \ref{main2d} and \ref{main3d} below) with those of \cite{BEL2} where the authors considered a 1-dimensional version of the same model. This can be summed up by comparing the two illustrations of Figure \ref{fig:phase2D} (representing the 2- and 3-dimensional models) with the final illustration within Figure 4 of \cite{BEL2} representing the 1-dimensional model. The overall picture is very similar; the differences are in the precise locations of the thresholds between staticity almost everywhere and takeover almost everywhere (around 0.3531 in the 1-dimensional case, 0.3652 in the 2-dimensional case, and 0.3897 in the 3-dimensional case), as well as the fact that in the 2- and 3- dimensional model we continue to have small grey areas around those thresholds where the outcome has not yet been rigorously established. This will be explained further in Section \ref{subsection:mainthms} below.

Following the online release of an earlier version of the current paper, Immorlica, Kleinberg, Lucier, and Zadomighaddam released an unpublished manuscript\footnote{N.\ Immorlica, R.\ Kleinberg, B.\ Lucier, M.\ Zadomighaddam Exponential Segregation in a Two-Dimensional Schelling Model with Tolerant Individuals, preprint, arXiv: 1511.02537.} which analyses a closely related, but interestingly distinct, two dimensional Schelling system. They show, under a continuous time Poisson clock dynamic, and with the hypothesis that agents of both types share the same intolerance which is close to, but less than $0.5$ (i.e. $\tau_\alpha = \tau_\beta = \frac{1 - \varepsilon}{2}$), that the size of the segregated regions which emerge will, with probability approaching 1, be of exponential size in the neighbourhood radius ($e^{\Theta(w^2)}$).

\subsection{{\bf The Schelling model across science}} \label{subsection:science}
While Schelling's ultimate concern was to understand some of the mechanisms underlying racial  segregation in the United States, such a model may also be applied to many contexts in which the `individuals' represent particles or agents of a sort which might more plausibly behave according to simple rules which govern their behaviour at the local level. In fact these model are rightly seen as fitting within a larger family  of discrete time models arising in diverse fields. Before we give a formal description of our model, we review two other important examples.

\subsubsection{Ising Models}
Many authors have pointed out direct links to the Ising model, used to analyse phase transitions in the context of statistical mechanics (see  \cite{SS,DM,PW,GVN,GO}). As this observation suggests, the dynamics can fruitfully be analysed by assigning an energy level to any given configuration, typically corresponding to some measure of the mixing of types. Typically, the Ising model is considered at temperature $T>0$, meaning that while transitions are more much likely to occur from configurations with higher energy to those with lower energy, they may occasionally also occur in the opposite direction. As discussed above, this corresponds to a perturbed Schelling model, and eases the analysis.

In models with an unperturbed dynamics (i.e. with $T=0$), which are the focus of the current paper, transitions which increase the energy never occur. Thus our results fit into the statistical physical theory of \emph{rapid cooling}.
 \begin{figure}
 \centering
 \includegraphics[width=1.3in]{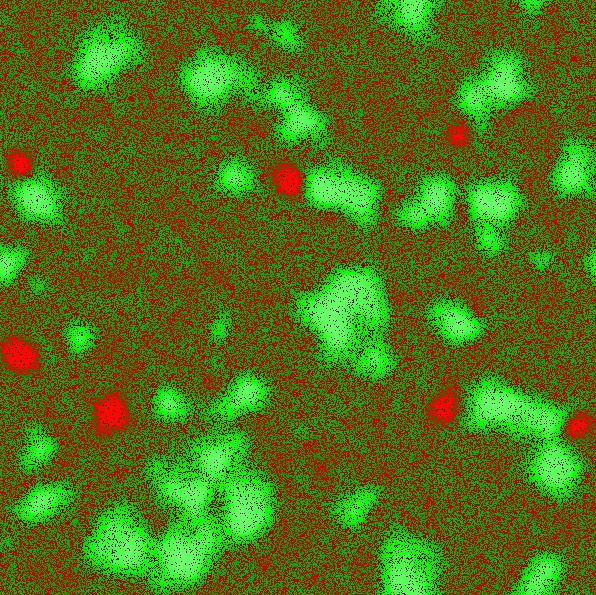}
 \includegraphics[width=1.3in]{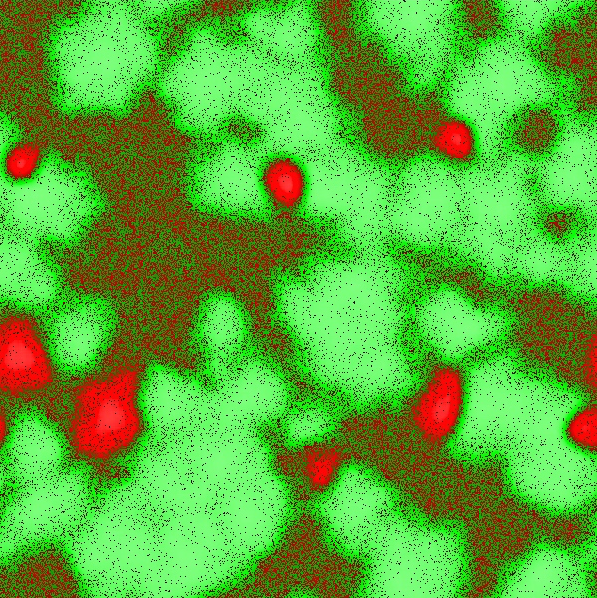}
  \includegraphics[width=1.3in]{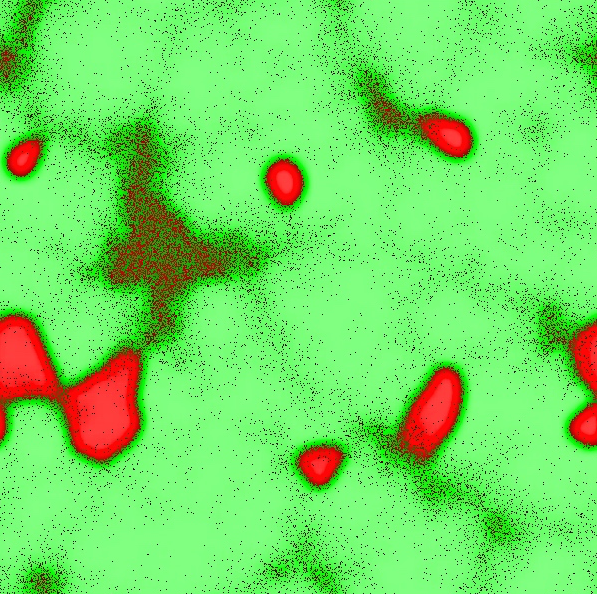}
   \includegraphics[width=1.3in]{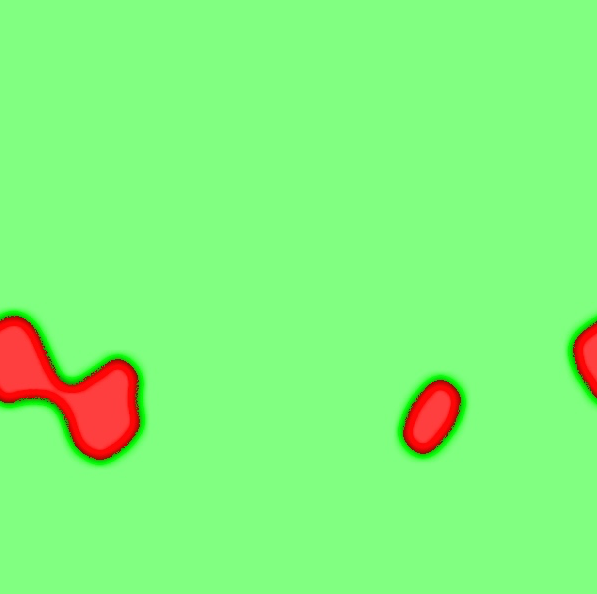}
\caption{The evolving process: $n=600,w=10,\tau_{\alpha}=0.44,\tau_{\beta}=0.42$.}
\label{fig:evolvingproc2}
\end{figure}

There is another important remark to make. In the current paper (as in \cite{BEL2}), we consider models in which individuals of different type may have unequal intolerances, $\tau_\alpha$ and $\tau_\beta$. As discussed in Section 1.3. of \cite{BEL2}, this equates to an Ising model with a non-standard external field whose strength differs for spins in state +1 spins compared to those in state -1. Assuming the usual notation for an Ising model, if we attempt to construct a Hamiltonian function in the standard way: 
$$H = - \sum_{||i-j||_{\infty} \leq w} \sigma_i \sigma_j + \sum_i \sigma_i K(\sigma_i)$$
we find that in general $|K(+1)| \neq |K(-1)|$, meaning that $H$ does not act as an energy function, since it may increase as well as decrease (even in the unperturbed case).

Thus, our results below can be understood in relation to an Ising model under rapid cooling, with range of interaction $w$ and a non-standard external field as described.

\subsubsection{Neural Networks} Hopfield networks provide an example of a model from the same family, in which the unperturbed dynamics are of particular interest. Introduced by Hopfield \cite{JH} in 1982 these are recurrent artificial neural networks, which have been much studied as a form of associative memory. Here, in fact, the unperturbed dynamics would seem to be an essential aspect of  the intended functionality---corresponding to any initial state the absorbing state reached within a finite number of steps of the dynamical process is the `memory' which the net is considered as associating with that input. Many authors have analysed the connections to spin-glass models \cite{AGS}. A major difference with the Schelling model though, is that the standard form of Hopfield network has no geometry: each threshold node takes inputs from all others. From a biological perspective, however, it is of interest to study models of sparse connectivity \cite{CAS,CN}.  The results of this paper can easily be adapted so as to correspond to a form of `local' Hopfield net, in which nodes arranged on a grid are only locally connected.

\subsubsection{Cascading Phenomena} Our own avenue into these questions came via connections to the study of cascading phenomena on networks as studied by Barab\'{a}si, Kleinberg and many others, a good introduction to which can be found in \cite{JK}.  The dynamics of the Schelling process are either identical or almost identical to versions of the Threshold Model used to model the flow of information, technology, behaviour, viruses, opinions etc.,  on large real world networks. Depending on context, both the perturbed and unperturbed dynamics may be of interest here. The principal difference with the Schelling model is typically underlying network, instead of a regular lattice, a random graph of some form (e.g. Watts-Strogatz graph \cite{WS}, or Barab\'{a}si-Albert preferential attachment \cite{BA}), which better reflects the clustering coefficient and degree distribution of the real world network in question. 

Given the results of this paper, an immediate question is whether the techniques developed here can be applied to understand emergent phenomena on such random structures.   Along these lines, Henry, Pra\l at and Zhang have described a simple but elegant model of network clustering \cite{HPZ}, inspired by Schelling segregation. Although an interesting system, it does not display the kind of involved threshold behaviour that one might expect; thus there may be more to be said here.

\begin{table}\begin{center}
\begingroup\setlength{\fboxsep}{7pt}
\colorbox{lightgray}{%
  \begin{tabular}{lcl}
 \hline\hline\\[1ex]
 \textrm{\footnotesize{\bf Static almost everywhere}}  &  & \parbox{6.8cm}{\textrm{\footnotesize With high probability a random node does not change its type throughout the process.}}  
 \\[1ex]\\[1ex]
\textrm{\footnotesize{\bf $\alpha$ (or $\beta$)  takeover almost everywhere}}  &  & \parbox{6.8cm}{\textrm{\footnotesize A random node has
high probability of being of type $\alpha$ ($\beta$ respectively)  in the final configuration.}}  
\\[1ex]\\[1ex]
\textrm{\footnotesize{\bf  $\alpha$  (or $\beta$) takeover totally}}  &  & \parbox{6.8cm}{\textrm{\footnotesize With high probability all 
nodes are of type $\alpha$ ($\beta$ respectively) in the final configuration.}}  \\[1ex]\\[1ex]
 \hline\hline\end{tabular}}\endgroup
  \vspace{0.4cm}
\caption{The three (or five) behaviours of the model.}
\label{ta:propofa}\end{center}
\end{table}

\subsection{{\bf The model}} \label{model def} We describe the two dimensional model first, and it is then simple to extend to three dimensions. For the two dimensional model we consider an $n\times n$ `grid' of nodes. To avoid boundary issues, we work on the flat torus $\mathbb{T}=[0,n-1)\times [0,n-1)$, i.e.\ $\mathbb{R}^2/\equiv$ where  $(x,y)\equiv (x+n,y) \equiv (x,y+n)$. Thus, in the context of discussing Euclidean coordinates, arithmetical operations are always performed modulo $n$. We let $\mathbb{R}_n$ and $\mathbb{N}_n$ denote the reals and the natural numbers modulo $n$ respectively.
Formally, then, `nodes' are  elements of $\mathbb{N}_n^2$. In the initial configuration (at stage $0$) each node is assigned one of two types, either $\alpha$ or $\beta$. In this initial configuration the types of nodes are independent and identically distributed, with each node having probability $0.5$ of being type $\alpha$.

For each node we also consider a certain \emph{neighbourhood}, with size specified by a parameter $w$: the neighbourhood of the node $\boldsymbol{u}$, denoted $\mathcal{N}(\boldsymbol{u})$  is the set of nodes $\boldsymbol{v}$ such that $||\boldsymbol{u}-\boldsymbol{v}||_{\infty} \leq w$. So far, then, we have specified two parameters for the model, $n$ and $w$. The two remaining parameters are the \emph{intolerance} levels $\tau_{\alpha},\tau_{\beta}\in [0,1]$. At any given stage we say that a node $\boldsymbol{u}$ of type $\alpha$ is \emph{happy} if the proportion of the nodes in $\mathcal{N}(\boldsymbol{u})$ which are of type $\alpha$ is at least $\tau_{\alpha}$ (note that $\boldsymbol{u}$ is included in its own neighbourhood). Similarly a node $\boldsymbol{u}$ of type $\beta$ is happy if the proportion of the nodes in $\mathcal{N}(\boldsymbol{u})$ which are of type $\beta$ is at least $\tau_{\beta}$. We say that a node is  \emph{hopeful} if it is not happy, but would be happy if it changed type (the types of all other nodes remaining unchanged).  
The dynamical process then unfolds as follows. At each stage $s+1$, we consider the set of hopeful nodes at the end of stage $s$. Taking each in turn we change the type, providing it is still hopeful given earlier changes during stage $s+1$. The process terminates when there are no remaining hopeful nodes.

 In the description of the dynamic process above, we have stated that during each stage $s+1$, each node which was hopeful at the end of stage $s$ is taken \emph{in turn}, the type then being changed if it still remains hopeful given earlier changes during stage $s+1$. For the sake of definiteness, we may consider the nodes to be lexicographically ordered, but this choice is essentially arbitrary, and it will be clear from what follows, that all results apply for any ordering of the nodes (since the particular form of the ordering is not used in the proof of any lemma). In fact the model will behave in an almost identical fashion if,  instead of taking the hopeful nodes in order at each stage, we instead pick hopeful nodes uniformly at random. This fact is easy verified by simulation, and our expectation is that it would not be difficult to use techniques applied in \cite{BEL2} to extend the results presented here to such a dynamics. The choice to consider nodes in order at each stage, simply means that certain aspects of the dynamics become easier to analyse. Roughly, this is  because we are immediately guaranteed that the process will unfold at the same rate at different locations on the grid,  rather than having to apply a probabilistic analysis to show that formal statements to this effect will be close to true most of the time. It \emph{is} important for the analysis, however, that the set of hopeful nodes considered during stage $s+1$, is fixed at the end of stage $s$ --  thereby avoiding `waves' of changes traversing large distances in a single stage.

It is easy to define an appropriate Lyapunov function, establishing that the process must eventually terminate. For example, one may consider the sum over all nodes $\boldsymbol{u}$,  of the number nodes of the same type as $\boldsymbol{u}$ in $\mathcal{N}(\boldsymbol{u})$.  
The three dimensional model is defined identically, except that nodes are now elements of  $\mathbb{N}_n^3$. Once again,  $\mathcal{N}(\boldsymbol{u})$  is the set of nodes $\boldsymbol{v}$ such that $||\boldsymbol{u}-\boldsymbol{v}||_{\infty} \leq w$, but obviously the $\ell_{\infty}$ metric is now understood  in three dimensions.

Note that the numbers of nodes of each type do not remain stable throughout the process. Although the initial expectation is 50\% of each type, during the process it is conceivable that one type may increase or decrease its population. In fact, intuitively one would expect that the type with the least intolerance level has an advantage over the other type, as it
is less susceptible to changes. As we are going to see in the next section, this is largely true but there are several caveats. For example if both intolerance levels are very low or very high, then one would expect the system to remain largely stable. In the following section we make precise statements about what can be rigorously proved, establishing  most of the phase diagram.
\begin{figure}
 \centering
\includegraphics[width=1.7in]{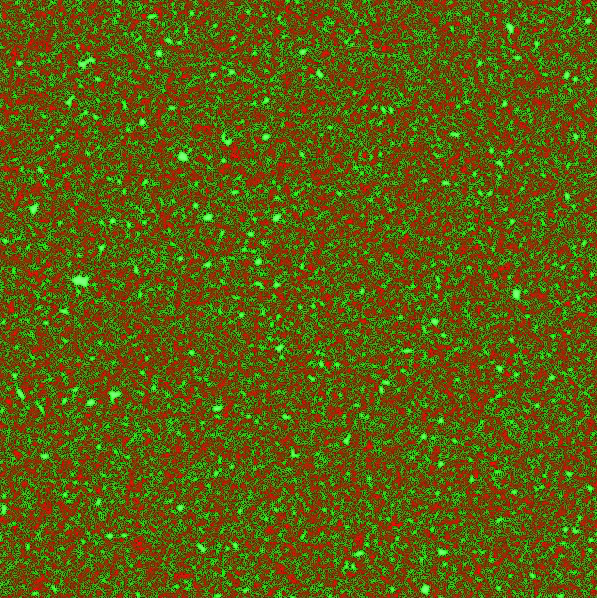} 
\includegraphics[width=1.7in]{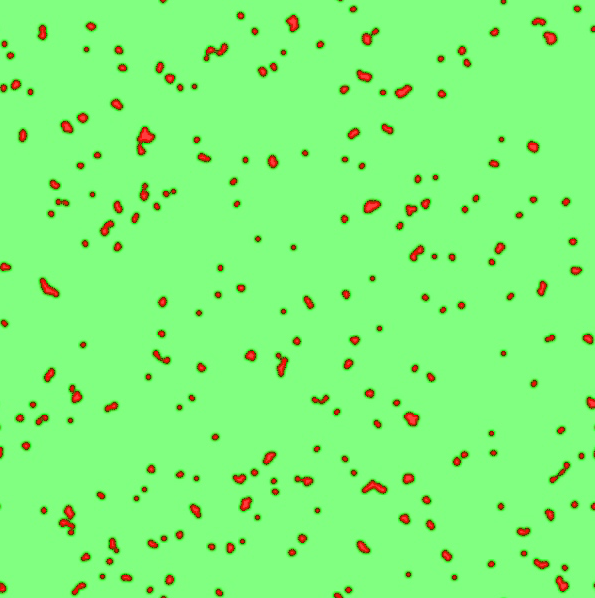}
\includegraphics[width=1.7in]{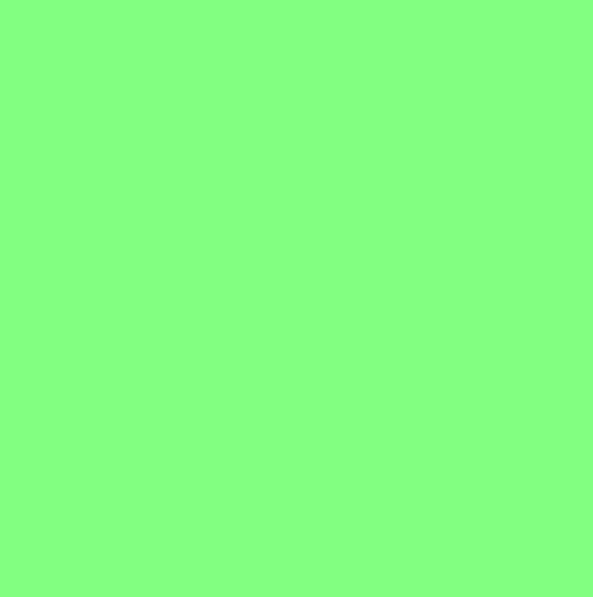}
\caption{Simulations showing three behaviors of the model, with $\alpha$ and $\beta$ nodes depicted as red and green respectively. From left to right we have ``static almost everywhere'',  ``$\beta$ takeover almost everywhere'' and ``$\beta$ takeover totally''. These are the final states of the model with $n=600$, $w=2$ and $(\tau_{\alpha}, \tau_{\beta})$ equal to (0.249,0.1), (0.4,0.3) and (0.6, 0.4) respectively.}
\label{fig:evolvscenario}
\end{figure}

\subsubsection{Behaviors of the model}
Our results are asymptotic in nature, and we will use the shorthand ``for $0 \ll w \ll n $'' to mean ``for all sufficiently large $w$, and all $n$ sufficiently large compared to $w$''. By a \emph{scenario} we mean the class of all instances of the model with fixed values of  $\tau_{\alpha}$, and $\tau_{\beta}$, but $w$ and $n$ varying. We will identify a scenario with its signature pair $(\tau_\alpha,\tau_\beta)$.

The phase diagram of the model will be expressed according to five behavior types, indicating the degree of growth of one population over the other that can occur by the time the model reaches its final state. Our results and analysis justify this natural choice of behavioural classification. Table \ref{ta:propofa}
shows the names we give to the different behaviors of the model, along with their corresponding descriptions. 
The precise definitions are given in Definition \ref{de:fivedifsc}.

\begin{defin}[Five different behaviors]\label{de:fivedifsc}
We list the three (or five, taking into account the two types) main ways in which a scenario can behave: 
\begin{enumerate}[(i)]
\item  A  scenario is \emph{static almost everywhere} if the following holds for every $\epsilon>0$: for $0\ll w \ll n$ a node $\boldsymbol{u}$ chosen uniformly at random has a probability $>1-\epsilon$ of having its type unchanged throughout the process.
\item  Type $\alpha$ ($\beta$) \emph{takes over almost everywhere} if the following holds for every $\epsilon>0$: for $0\ll w \ll n$ a node $\boldsymbol{u}$ chosen uniformly at random has a probability $>1-\epsilon$ of being of type $\alpha$ ($\beta$) in the final configuration.
\item Type  $\alpha$ ($\beta$) \emph{takes over totally} if the following holds for every $\epsilon>0$: for $0\ll w\ll n$ the probability that all nodes are of type $\alpha$ ($\beta$) in the final configuration exceeds $1-\epsilon$.
\end{enumerate}
\end{defin}

Figure \ref{fig:evolvscenario} shows the final states of the model for three different choices of signature 
$(\tau_{\alpha},\tau_{\beta})$. The first state indicates a ``static almost everywhere'' behavior, with only occasional
areas where nodes have switched to type $\beta$. The second state indicates an 
``almost everywhere $\beta$ takeover'', with the occasional small areas where nodes have succeeded in keeping the
$\alpha$-type. Finally the third state is a total $\beta$ takeover.
The signatures for the corresponding behaviors illustrate Theorem \ref{main2d} of the next section.

In general, when $\tau_{\alpha}<\tau_{\beta}$, we might expect the $\alpha$ type to be more persistent
in the process. In order to formally establish results expressing such expectations, we often need a stronger
notion of inequality, expressing that  (for example) $\tau_{\alpha}$ is {\em sufficiently less than} $\tau_{\beta}$.
In the case of the 2-dimensional model we denote this relation by $\tau_{\alpha} \lhd \tau_{\beta}$ (the
formal definition is given in Definition \ref{definition:suf2Dine}). 
In the case of the 3-dimensional model sufficient inequality is a slightly stronger condition, and is denoted by
$\tau_{\alpha} \leftslice \tau_{\beta}$ (the formal definition is given in Definition \ref{de:suff3D}). Note that 
`sufficiently less' has different meanings for  the 2-dimensional and 3-dimensional models. This will not
cause confusion as the two models are dealt with separately, and we use different symbols for the two notions of 
inequality. The second item of Figure \ref{fig:exfigaun} illustrates the relation between sufficient inequality for 
the  2-dimensional and 3-dimensional models with the usual inequality.

\subsubsection{Main theorems} \label{subsection:mainthms}
Our main result for the two dimensional model is the following classification of behaviors of the model,
according to the signature $(\tau_{\alpha}, \tau_{\beta})$.
The different cases correspond to the positions of the two intolerance levels with respect to each other
and other values like 0.5, 0.25 and a certain constant $\kappa \approxeq 0.365227$ that will be formally 
defined (as the unique solution of a certain equation) shortly.
Recall that $\lhd$ indicates the `sufficiently less' relation for the two dimensional model 
(see Definition \ref{definition:suf2Dine}). 
Given the symmetric roles of $\alpha$ and $\beta$, it is clear that the following theorem (and all those which follow)  also holds with the roles of $\alpha$ and $\beta$ reversed. From now on we shall not explicitly mention this fact, and will leave it to the reader to fill in these symmetries.  
\begin{figure}
 \centering
\includegraphics[width=2.9in]{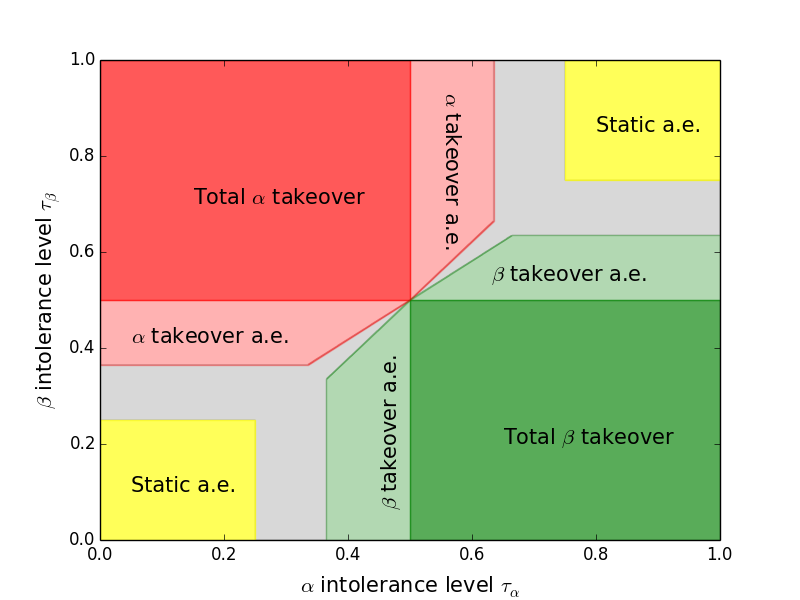} 
\includegraphics[width=2.9in]{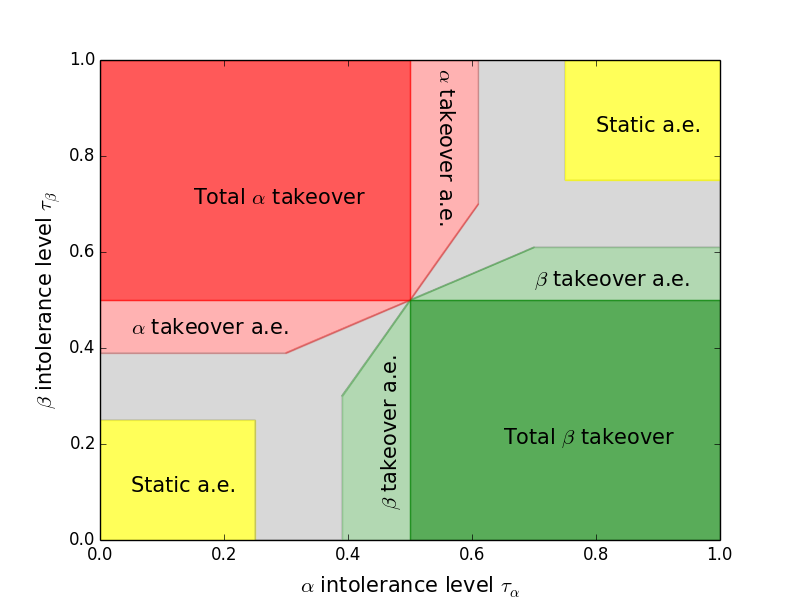} 
\caption{The phase diagrams of the two and the three dimensional models respectively. Note that in the latter we have a larger grey region, corresponding to behaviour which we have not been able to rigorously establish.}
\label{fig:phase2D}
\end{figure}

\begin{thm}[Behavior of the two dimensional model]  \label{main2d}
The behavior of the two dimensional model is dictated by the signature $(\tau_{\alpha}, \tau_{\beta})$ as follows:
\begin{enumerate}[\hspace{0.5cm}(a)]
\item  If $\tau_\alpha, \tau_\beta <0.25$ then the scenario is static almost everywhere. 
\item   If $\kappa <\tau_{\alpha}<0.5$ and $\tau_{\beta} \lhd \tau_{\alpha}$ then 
 $\beta$ takes over almost everywhere. 
\item For $\tau_{\beta}<0.5 <\tau_{\alpha}$,  $\beta$ takes over totally.
\item If $1-\kappa >\tau_{\alpha}>0.5$ and $\tau_{\beta} \vartriangleright \tau_{\alpha}$ 
then  $\alpha$ takes over almost everywhere. 
\item If $\tau_{\alpha},\tau_{\beta}>0.75$ then the scenario is static almost everywhere.
 \end{enumerate}
where $\kappa$ is the unique solution in $[0,1]$ to  
$(1-2\kappa)^{1-2\kappa} = 2^{2(1-\kappa)}\kappa^\kappa (1-\kappa)^{3(1-\kappa)}$ 
(numerically $\kappa \approxeq 0.365227$). 
 \end{thm}

So if $\tau_{\alpha}$ and $\tau_{\beta}$ are both sufficiently small then most nodes remain unchanged throughout the process. If at least one of $\tau_{\alpha}$ and $\tau_{\beta}$ is reasonably large while remaining in the interval $(0,0.5)$,  however, then the situation changes dramatically. Here `reasonably large' means above the threshold $\kappa$.  While it may not be immediately obvious where the equation defining $\kappa$ comes from, it will be derived later by comparing the probabilities of certain structures in the initial configuration.

 Figures \ref{fig:evolvingproc} and \ref{fig:evolvingproc2} display the evolving process for simulations corresponding to the scenario of clause (b) of Theorem \ref{main2d}.\footnote{The C++ code for these simulations is available at http://barmpalias.net/schelcode.shtml.}  Here nodes of type $\alpha$ are depicted red, while those of type $\beta$ are green. The shade of red or green corresponds to the number of same-type nodes within $\mathcal{N}(\boldsymbol{u})$: the brighter the shade the more nodes there are of the same type as $\boldsymbol{u}$ in $\mathcal{N}(\boldsymbol{u})$. Since   $\tau_{\beta} \lhd \tau_{\alpha}$, most nodes are of type $\beta$ in the final configuration, with this proportion tending to 1 as $w\rightarrow \infty$. 
 
Using almost identical proofs, we can get slightly weaker results for the three dimensional model. 

\begin{thm}[Behavior of the three dimensional model]  \label{main3d}
The behavior of the three dimensional model is dictated by the signature $(\tau_{\alpha}, \tau_{\beta})$ as follows:
\begin{enumerate}[\hspace{0.5cm}(a)]
\item If $\tau_\alpha, \tau_\beta <0.25$ then the scenario is static almost everywhere. 
\item  If $\kappa_{\ast} <\tau_{\alpha}<0.5$ and $\tau_{\beta} \leftslice \tau_{\alpha}$ 
then  $\beta$ takes over almost everywhere.
\item For $\tau_{\beta}<0.5 <\tau_{\alpha}$,  then $\beta$ takes over totally.
\item If $1-\kappa_{\ast} >\tau_{\alpha}>0.5$ and $\tau_{\beta} \rightslice \tau_{\alpha}$ 
then  $\alpha$ takes over almost everywhere. 
\item   If $\tau_{\alpha},\tau_{\beta}>0.75$ then the scenario is static almost everywhere. 
\end{enumerate}
where $\kappa_{\ast}$ is the unique solution in $[0,1]$ to
$(1-2\kappa_{\ast})^{4(1-2\kappa_{\ast})} = 
2^{23-8\kappa_{\ast}}\kappa_{\ast}^{19\kappa_{\ast}} (1-\kappa_{\ast})^{27(1-\kappa_{\ast})}$
(numerically, $\kappa_{\ast} \approxeq 0.3897216$).
\end{thm} 

Note that our results for the three dimensional model are the same as for the two dimensional model 
except that the thresholds are slightly different (making the gaps in the phase diagram slightly larger).

\subsection*{{\bf Sufficient inequality for the 2D and 3D models}}
In this section we formally define the notions of sufficient inequality (denoted by $\lhd$ and $\leftslice$)
for the two dimensional and the three dimensional models respectively.
The particulars of the definition result from  comparing the probabilities of certain structures 
in the initial configuration, and the motivation behind the definition will become clear in Section \ref{section:two}. 

\begin{defin} [Sufficient inequality for 2D] \label{definition:suf2Dine}
We let $g(x,k)=x^{kx}(1-x)^{k(1-x)}$. For $\tau_0,\tau_1 \in (0,0.5)$ we say that $\tau_0$ is 
sufficiently less than $\tau_1$, denoted $\tau_0 \lhd \tau_1$,  if  $g(\tau_0,2)>2g(\tau_1,3)$. 
For $\tau_0,\tau_1 \in (0.5,1)$ we say that $\tau_0$ is sufficiently greater 
than $\tau_1$, denoted $\tau_0 \vartriangleright \tau_1$,  if  $1-\tau_0$ is sufficiently less than $1-\tau_1$. 
\end{defin}

\begin{figure}
 \centering
\includegraphics[width=3in]{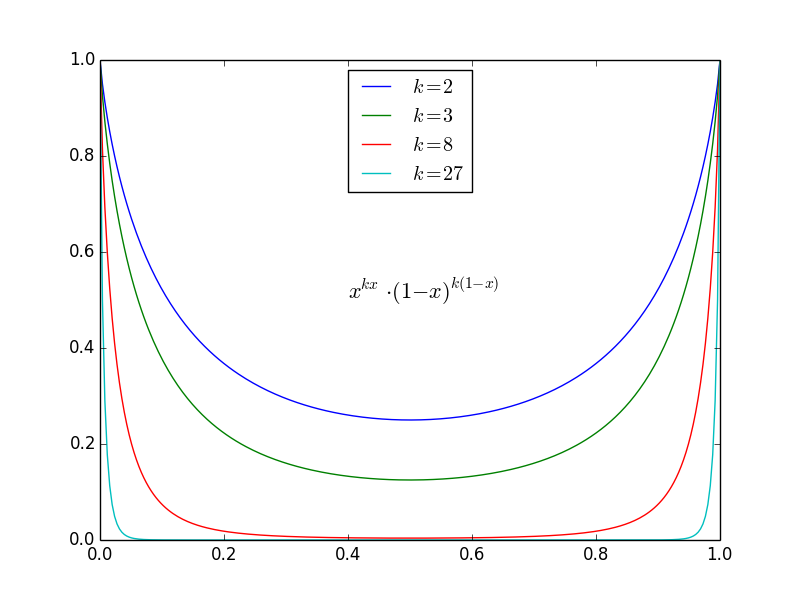} 
\includegraphics[width=3in]{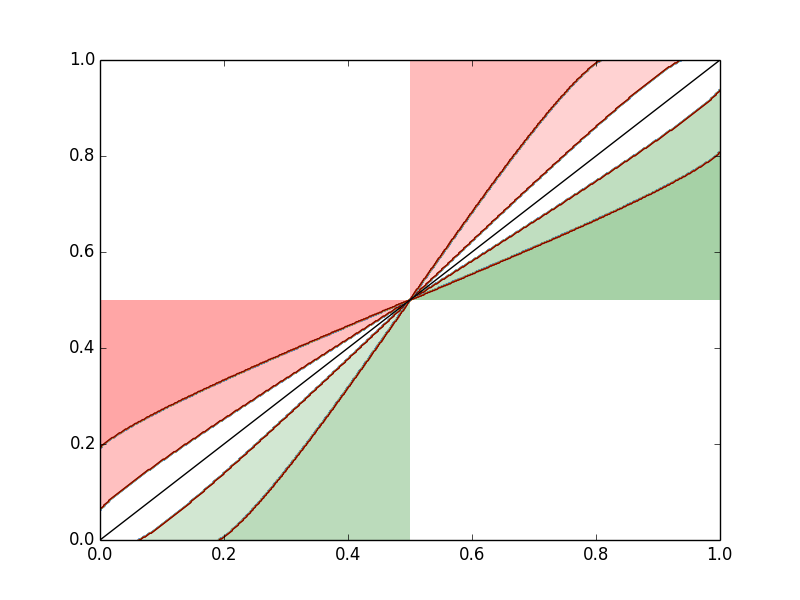} 
\caption{The function $g$ of Definitions \ref{definition:suf2Dine} and \ref{de:suff3D}, for the values of $k$ that we use.
In the second figure the horizontal axis is $\tau_\alpha$ and the perpendicular axis is $\tau_\beta$. Then the 
three regions show the relations $\tau_\beta \leftslice\tau_\alpha$, $\tau_\beta \lhd \tau_\alpha$,
 $\tau_\beta <\tau_\alpha$ in $[0,0.5]$ as well as
 $\tau_\beta \rightslice\tau_\alpha$, $\tau_\beta \rhd \tau_\alpha$,
 $\tau_\beta >\tau_\alpha$ in $[0.5, 1]$.}
\label{fig:exfigaun}
\end{figure}

Note that the function $g$ is decreasing in $[0, 0.5]$ and increasing in $[0.5,1]$ (see the first item of Figure
\ref{fig:exfigaun}). Hence $\lhd$ in $[0, 0.5]$ is a strengthening of $<$ and $\rhd$ in $[0.5,1]$ is a strengthening of $>$.
The requirement that  $\tau_{\beta} \lhd \tau_{\alpha}$ is not overly strong, as the second item of Figure \ref{fig:exfigaun}   shows. Table \ref{ta:prodsfa} is a chart that shows for various values of $\tau_\alpha$ in $[0,0.5]$ the amount by which
$\tau_\beta$ needs to be smaller in order for $\tau_\beta\lhd \tau_\alpha$ to hold (and similarly for $\leftslice$ of Definition \ref{de:suff3D}).

\begin{defin}[Sufficient inequality for 3D] \label{de:suff3D}
Again let $g(x,k)=x^{kx}(1-x)^{k(1-x)}$. For $\tau_0,\tau_1 \in (0,0.5)$ we say that $\tau_0$ is sufficiently less than
$\tau_1$, denoted $\tau_0 \leftslice\tau_1$,  if  $g(\tau_0,8)>2^{19}g(\tau_1,27)$. 
For $\tau_0,\tau_1 \in (0.5,1)$ we say that $\tau_0$ is sufficiently greater than $\tau_1$, 
denoted $\tau_0 \rightslice \tau_1$,  if  $1-\tau_0$ is sufficiently less than $1-\tau_1$.   
\end{defin}

The relation $\leftslice$ is significantly stronger 
than $\lhd$, as is made clear by  Figure \ref{fig:exfigaun} and Table \ref{ta:prodsfa}. 

\subsection{{\bf Further notation and terminology}} 
The variables $\boldsymbol{u}$ and $\boldsymbol{v}$ are used to range over nodes, while $r, x,y,z, \tau$  range over $\mathbb{R}$. We let  $\boldsymbol{x}$ and $\boldsymbol{y}$ range over either $\mathbb{R}^2$ or $\mathbb{R}^3$ depending on context, while   $a,b,c,d,e,i,j,$ $k,n,m,w$ range over $\mathbb{N}$. Often we shall use $x$ and $y$ (and $z$ when working in three dimensions) to specify the coordinates of a node $\boldsymbol{u}=(x,y)$ -- so while $x,y$ and $z$ range over $\mathbb{R}$, often they will be elements of $\mathbb{N}$. 

 In the context of discussing a given scenario $(\tau_{\alpha},\tau_{\beta})$, we say that an event $X$ occurs \emph{in the limit}, if for every $\epsilon>0$, $X$ occurs with probability $>1-\epsilon$ for $0\ll w \ll n$, i.e.\ $X$ holds with probability $>1-\epsilon$ whenever $w$ is sufficiently large and $n$ is sufficiently large compared to $w$, and where it is to be understood that how large one has to take $w$ and how large $n$ must be in comparison to $w$ may depend upon $\epsilon, \tau_{\alpha}$ and $\tau_{\beta}$. 
When working in two dimensions we let $C_{\boldsymbol{u},r}$ be the circle of radius $r$ centred at $\boldsymbol{u}$, and we let $C^{\dagger}_{\boldsymbol{u},r}$ be the disc of radius $r$ centred at $\boldsymbol{u}$, i.e.\ the set of all $\boldsymbol{x}\in \mathbb{T}$ with $|\boldsymbol{u}-\boldsymbol{x}|\leq r$ (where $|\boldsymbol{u}-\boldsymbol{x}|$ denotes the Euclidean distance between $\boldsymbol{u}$ and $\boldsymbol{x}$). Then we extend this notation to $A\subseteq \mathbb{T}$ as follows. If $A\subset C^{\dagger}_{\boldsymbol{u},r}$ then we define $A^{\dagger}$ \emph{with respect to} $C^{\dagger}_{\boldsymbol{u},r}$ to be the intersection of all convex subsets of $C^{\dagger}_{\boldsymbol{u},r}$ containing $A$. If $r<0.25n$ and $A\subset C^{\dagger}_{\boldsymbol{u},r}$, then we shall call $A$ \emph{local}. If $A$ is local then  $A^{\dagger}$ is the same with respect to all $C^{\dagger}_{\boldsymbol{u},r}$ such that $r<0.25n$ and  $A\subset C^{\dagger}_{\boldsymbol{u},r}$, and we shall suppress mention of $C^{\dagger}_{\boldsymbol{u},r}$ in this case, referring simply to $A^{\dagger}$.  In the context of discussing $A$ which is local (such as $\mathcal{N}(\boldsymbol{u})$ for a node $\boldsymbol{u}$) and two nodes $\boldsymbol{u}=(x,y)$ and $\boldsymbol{v}=(x',y')$ in $A$, we say that $x\leq x'$ if there exists $r<0.5n$ with $x+r =x'$ (and similarly for $y$ and $y'$).  By an $m$-square we mean a subset of $\mathbb{N}_n^2$ of the form $[x,x+m) \times [y,y+m)$ for some $x,y\in \mathbb{N}$. By an $a\times b$ rectangle we mean a subset of $\mathbb{N}_n^2$ of the form $[x,x+a) \times [y,y+b)$ for some $x,y,a,b \in \mathbb{N}$.

\subsection{{\bf Overview of the proof}} \label{overview} At the start of Section \ref{section:two} we shall prove Theorem \ref{main2d} (a), which along with Theorem \ref{main2d} (e) is the easiest case and amounts to little more than some easy observations. The bulk of the work will then be in the remainder of Section \ref{section:two}, in which we prove  Theorem \ref{main2d} (b) and Theorem \ref{main2d} (c). With these in place, Theorem \ref{main2d} (d)  and  Theorem \ref{main2d} (e)  will follow from certain symmetry considerations, as described in Section \ref{upper half}.  In Section \ref{sec five} we describe how to modify the proofs for the two dimensional model, to give all of the corresponding theorems for the three dimensional model.  Throughout, certain proofs which are of a technical nature, and which the reader might profitably skip on a first reading, are deferred to Section \ref{defer}. In the remainder of this subsection, we outline some of the basic ideas behind the proof of Theorem \ref{main2d}. 

The key to the proof lies in considering certain local structures, which might exist in the initial configuration or else might develop during the dynamic process. Our first central notion is that of a \emph{stable} structure: 

\begin{defin}[Stable structures]
We say that a set of nodes $A$ is an $\alpha$-\emph{stable structure} if, for every $\boldsymbol{u}\in A$, there are at least $\tau_{\alpha}(2w+1)^2$ many $\alpha$ nodes in $\mathcal{N}(\boldsymbol{u}) \cap A$.  We define $\beta$-stable structures analogously.
\end{defin} 
The point of an $\alpha$-stable structure is this (as may be seen by induction on stages):  no node of type $\alpha$ which belongs to an $\alpha$-stable structure can ever change type. We shall also be interested in certain conditions on $\mathcal{N}(\boldsymbol{u})$ and variants of this neighbourhood:

\begin{itemize}
\item (Partial neighbourhoods.) Suppose that $\ell$ is a straight line passing through $\mathcal{N}^{\dagger}(\boldsymbol{u})$. Let $A_1$ be all those points in $\mathcal{N}^{\dagger}(\boldsymbol{u})$ on or above $\ell$, and let $A_2$ be all those points in $\mathcal{N}^{\dagger}(\boldsymbol{u})$ on or below $\ell$. If $\mu(A_i^{\dagger})=\gamma(2w+1)^2$ ($\mu$ denotes Lebesgue measure), then we call $A_i$ a $\gamma$-partial neighbourhood of $\boldsymbol{u}$, with defining line $\ell$.   
If $\boldsymbol{u}$ is an $\alpha$ node,  we let $\mathtt{pn}^{\alpha}_{\gamma,\tau}(\boldsymbol{u})$ be the event that there exists some $\gamma$-partial  neighbourhood of $\boldsymbol{u}$ which  contains at least $\tau(2w+1)^2$ many $\alpha$ nodes in the initial configuration (and similarly for $\beta$). We also let $\mathtt{pn}^{\alpha}_{\gamma,\tau}(\boldsymbol{u})[s]$ be the corresponding event for the end of stage $s$, rather than the initial configuration.\footnote{One should think of  $\mathtt{pn}$ as p-artial n-eighbourhood, $\mathtt{uh}$ as u-nh-appy, and $\mathtt{ruh}$ as r-ight-extended neighbourhood u-nh-appy.}

\begin{table}\begin{center}
\begingroup\setlength{\fboxsep}{7pt}
\colorbox{lightgray}{%
  \begin{tabular}{ccc}
 \hline\hline
 $\tau_{\alpha}$  & $\tau_{\beta} \lhd \tau_{\alpha}$ &  $\tau_{\beta} \leftslice \tau_{\alpha}$   
 \\[1ex]\hline
{\footnotesize  0.390000} & {\footnotesize 0.024443}& {\footnotesize 0.090076} \\
{\footnotesize  0.400000} & {\footnotesize 0.022266} & {\footnotesize 0.082213}\\
{\footnotesize 0.410000} & {\footnotesize 0.020076}& {\footnotesize 0.074255} \\
{\footnotesize 0.420000} & {\footnotesize 0.017874} & {\footnotesize 0.066211}\\[1ex]
 \hline\hline\end{tabular}
 \hspace{0.5cm}
  \begin{tabular}{ccc}
 \hline\hline
 $\tau_{\alpha}$  & $\tau_{\beta} \lhd \tau_{\alpha}$ &  $\tau_{\beta} \leftslice \tau_{\alpha}$   
 \\[1ex]\hline
{\footnotesize 0.430000} & {\footnotesize 0.015661} & {\footnotesize 0.058092}\\
{\footnotesize 0.450000} & {\footnotesize 0.011212} & {\footnotesize 0.041673}\\
{\footnotesize 0.470000} & {\footnotesize 0.006737}& {\footnotesize 0.025074} \\
{\footnotesize 0.490000} & {\footnotesize 0.002247}& {\footnotesize 0.008370} \\[1ex]
 \hline\hline\end{tabular}}
 \endgroup
  \vspace{0.4cm}
\caption{This chart shows, for various values of $\tau_{\alpha}$ in $[0,0.5]$, the minimum amount by which
$\tau_\beta$ needs to be less than $\tau_\alpha$ in order to have $\tau_{\beta} \lhd \tau_{\alpha}$ or
  $\tau_{\beta} \leftslice \tau_{\alpha}$ respectively.}
\label{ta:prodsfa}\end{center}
\end{table}

\item (Extended neighbourhoods.) We say that $\mathtt{uh}^{\alpha}_{\tau}(\boldsymbol{u})$ holds if there are strictly less than $\tau(2w+1)^2$ many $\alpha$ nodes in $\mathcal{N}(\boldsymbol{u})$ in the initial configuration (and again we let  $\mathtt{uh}^{\alpha}_{\tau}(\boldsymbol{u})[s]$ denote the corresponding event for the end of stage $s$). If  $\boldsymbol{u}=(x,y)$ then the \emph{right extended neighbourhood} of $\boldsymbol{u}$, denoted $\mathcal{N}^{\diamond}(\boldsymbol{u})$, is the set of nodes $(x',y')$ such that $x-w \leq x' \leq x+2w$ and $|y-y'|\leq w$. We say that  $\mathtt{ruh}^{\alpha}_{\tau}(\boldsymbol{u})$  holds if there are strictly less than $\tau(2w+1)(3w+1)$ many $\alpha$ nodes in $\mathcal{N}^{\diamond}(\boldsymbol{u})$ in the initial configuration.
\end{itemize}

So a $\gamma$-partial neighbourhood of $\boldsymbol{u}$ is a subset of $\mathcal{N}^{\dagger}(\boldsymbol{u})$ (of a particularly simple form) which has measure $\gamma(2w+1)^2$. If $\mathtt{pn}^{\alpha}_{\gamma,\tau_{\alpha}}(\boldsymbol{u})$ holds, then it is not only the case that $\boldsymbol{u}$ is happy -- it has some $\gamma$-partial neighbourhood which already contains enough $\alpha$ nodes for $\boldsymbol{u}$ to be happy. The right extended neighbourhood  $\mathcal{N}^{\diamond}(\boldsymbol{u})$ is a (particularly simple) superset of  $\mathcal{N}(\boldsymbol{u})$. While it is not strictly true that for $\tau<0.5$, $\mathtt{ruh}^{\alpha}_{\tau}(\boldsymbol{u})$ implies  $\mathtt{uh}^{\alpha}_{\tau}(\boldsymbol{u})$, it is not difficult to see that the former condition is less likely than the latter. 

While the notion of a stable structure might seem essentially passive, we can strengthen the notion to give a form of stable structure which will tend to grow over stages: 

\begin{defin}[Firewalls] \label{rast} 
If all nodes in $C^{\dagger}_{\boldsymbol{u},rw}$ are of type $\beta$ (at some stage) then this set of nodes is called a $\beta$-firewall of radius $rw$ centred at $\boldsymbol{u}$ (at that stage). We say $\boldsymbol{v}$ is on the \emph{outer boundary} of   $C^{\dagger}_{\boldsymbol{u},rw}$ if  $\boldsymbol{v}\notin C^{\dagger}_{\boldsymbol{u},rw}$ but has an immediate neighbour belonging to this set, i.e.\ there exists $\boldsymbol{v}'$ with $||\boldsymbol{v}-\boldsymbol{v}'||_{\infty} =1$ and $\boldsymbol{v}'\in C^{\dagger}_{\boldsymbol{u},rw}$. 
\end{defin}

Now suppose that $\tau_{\alpha},\tau_{\beta} <0.5 $, that $\gamma >0.5$ is fixed,  and that we choose some $r_{\ast}\in \mathbb{N}^+$ sufficiently large that, for $0\ll w \ll n$:
\begin{enumerate}[] 
\item $(\dagger_a)$ Any node $\boldsymbol{v}\in C^{\dagger}_{\boldsymbol{u},r_\ast w}$ satisfies the condition that $|\mathcal{N}(\boldsymbol{v}) \cap C^{\dagger}_{\boldsymbol{u},r_{\ast}w}| \geq \tau_{\beta}(2w+1)^2$.
\item $(\dagger_b)$ Any node $\boldsymbol{v}$ on the outer boundary of any $C^{\dagger}_{\boldsymbol{u},r_\ast w}$ has all of $\mathcal{N}(\boldsymbol{v})-C^{\dagger}_{\boldsymbol{u},r_\ast w}$ contained in some $\gamma$-partial neighbourhood of $\boldsymbol{v}$. 
\end{enumerate} 

We may now observe that if $r\geq r_{\ast}$ then for $0\ll w \ll n$: a) any $\beta$-firewall of radius $rw$ is a $\beta$-stable structure, and b)  any $\alpha$ node $\boldsymbol{v}$ on the outer boundary of a $\beta$-firewall of radius $rw$ at stage $s$ will be hopeful,  so long as  $\mathtt{pn}^{\alpha}_{\gamma,\tau_{\alpha}}(\boldsymbol{v})[s]$ does not hold.  For b), our choice of $r_\ast$ implies that $\boldsymbol{v}$  will be unhappy, and then since $\tau_{\alpha},\tau_{\beta}<0.5$ unhappiness automatically implies being hopeful. The point is this:

\begin{quote} 
A $\beta$-firewall of sufficient radius will grow over stages, so long as no node $\boldsymbol{v}$ on the outer boundary satisfies $\mathtt{pn}^{\alpha}_{\gamma,\tau_{\alpha}}(\boldsymbol{v})[s]$. 
\end{quote}

With these ideas in place, we can now sketch our approach to the  proof of Theorem \ref{main2d}. First of all suppose $\tau_{\alpha},\tau_{\beta}<0.25$ and consider Theorem \ref{main2d} (a). The basic idea is to show that for $0\ll w \ll n$,  a node $\boldsymbol{u}_0$ chosen uniformly at random very probably belongs to a structure which is both $\alpha$-stable and $\beta$-stable in the initial configuration. From our previous observations, it then follows that no nodes within this structure can ever change type. In fact a very simple choice of structure suffices -- if we fix $r$ which is sufficiently large then we can show that, for $0 \ll w \ll n$, $C^{\dagger}_{\boldsymbol{u}_0,rw}$ will very probably be such a stable structure. This follows because all nodes within the disc have (at least) close to half of their neighbourhood within the disc, and since (via an application of the weak law of large numbers) we can expect nodes within the disc to be quite evenly distributed. 

Now suppose that $\tau_{\alpha}>0.25$, $\tau_{\beta}<\tau_{\alpha}<0.5$ and consider Theorem \ref{main2d} (b). In this case it still holds that unhappy nodes can be expected to be rare in the initial configuration, and since $\tau_{\beta}<\tau_{\alpha}$ we have that unhappy $\beta$ nodes are less likely than unhappy $\alpha$ nodes. Taking $w$ large we have that unhappy $\beta$ nodes are \emph{much} less likely than unhappy $\alpha$ nodes. A naive approach might then proceed as follows. Choosing $\boldsymbol{u}_0$ uniformly at random, for $0\ll w \ll n$ we should be able to choose a large region around $\boldsymbol{u}_0$, $\mathcal{Q}$ say, which can be expected to contain unhappy $\alpha$ nodes but no unhappy $\beta$ nodes in the initial configuration. 
What then can be expected to occur in the vicinity of an unhappy $\alpha$ node $\boldsymbol{u}$ in the early stages? Well if  $\boldsymbol{u}$  changes type then this may cause other $\alpha$ nodes in $\mathcal{N}(\boldsymbol{u})$ to become unhappy. If these then change type then this may cause \emph{further} $\alpha$ nodes to become unhappy, and so on. In this manner a cascade of changes to type $\beta$ may emanate out from the (rare) initially unhappy $\alpha$ nodes, bringing about the formation of large $\beta$-firewalls. We might then look to establish that $\boldsymbol{u}_0$ very probably belongs to such a $\beta$-firewall  in the final configuration (and so must ultimately be of type $\beta$). There are two immediate difficulties with this initial plan, however. 

\begin{enumerate} 
\item First of all, it is not actually clear whether or not unhappy $\alpha$ nodes are likely to give rise to the formation of large $\beta$-firewalls. It is for \emph{this} reason that we consider the events $\mathtt{ruh}^\alpha_{\tau_{\alpha}}(\boldsymbol{u})$ and the condition $\tau_{\beta}\lhd \tau_{\alpha}$. The fact that $\tau_{\beta}\lhd \tau_{\alpha}$ is precisely what we need in order to be able to  choose  $\mathcal{Q}$ so that in the limit  there will be  
 $\alpha$ nodes $\boldsymbol{u}$ in $\mathcal{Q}$ for which $\mathtt{ruh}^\alpha_{\tau_{\alpha}}(\boldsymbol{u})$ holds in the initial configuraiton-- a less probable condition than $\mathtt{uh}^\alpha_{\tau_{\alpha}}(\boldsymbol{u})$ -- while maintaining the absence of unhappy $\beta$ nodes in $\mathcal{Q}$. When $\mathtt{ruh}^\alpha_{\tau_{\alpha}}(\boldsymbol{u})$ holds we \emph{are} able to show that a $\beta$-firewall of sufficient radius around $\boldsymbol{u}$ very probably results. 
 
 \item The second problem is that we aren't  guaranteed that such a $\beta$-firewall will spread until $\alpha$-firewalls interfere with the process --  all we observed above was that  if $r\geq r_{\ast}$ (with $r_{\ast}$ chosen appropriately) then for $0\ll w \ll n$, any $\alpha$ node $\boldsymbol{v}$ on the outer boundary of a $\beta$-firewall of radius $rw$ at stage $s$ will be hopeful,  \emph{so long as}  $\mathtt{pn}^{\alpha}_{\gamma,\tau_{\alpha}}(\boldsymbol{v})[s]$ does not hold. This is where $\kappa$ comes into play.  From the fact that $\tau_{\alpha}>\kappa$, we are able to choose an appropriate $\gamma>0.5$ and show that one can choose $\mathcal{Q}$ so as to ensure the absence of nodes in $\mathcal{Q}$ for which $\mathtt{pn}^{\alpha}_{\gamma,\tau_{\alpha}}(\boldsymbol{v})$ holds (in the initial configuration). We can then establish that our large $\beta$-firewall, formed within $\mathcal{Q}$ in the early stages of the process, will spread until $\boldsymbol{u}_0$ is contained within it. 
 \end{enumerate}
 
 Proving Theorem \ref{main2d}  parts (c), (d) and (e) then only involves simple modifications of the proofs for parts (a) and (b). 

\section{The proof of clauses (a), (b) and (c) of Theorem \ref{main2d}} \label{section:two}
\subsubsection*{The proof of Theorem \ref{main2d} (a)}  
As mentioned previously, proving Theorem \ref{main2d} (a) requires only some simple observations.   The basic idea is to show that, under the hypothesis of the theorem,  a node chosen uniformly at random will very probably belong to both $\alpha$-stable and $\beta$-stable structures in the initial configuration, and can therefore never change type.  

Throughout this subsection we assume that $\tau_\alpha,\tau_\beta<\frac{1}{4}$  are fixed and we work for varying $w$ and $n$ such that $0\ll w\ll n$. Choose $\tau_1,\tau_2,\tau_3$ such that  $\frac{1}{4}>\tau_1>\tau_2> \tau_3 >  \mbox{max} \{ \tau_{\alpha},\tau_{\beta} \}$.
We consider a fixed node $\boldsymbol{u}_0$, which is chosen uniformly at random.   We first observe that all nodes in  a large disc centred at $\boldsymbol{u}_0$ (with radius some multiple of $w$), will have (at least) close to half of their neighbourhood in the disc. Then the aim is to show that for large $w$ we can expect the nodes of each type to be distributed very evenly in the disc, meaning that nodes in the disc  can be expected to have (at least) close to $0.25(2w+1)^2$ many $\alpha$ nodes within the intersection of the disc and their neighbourhood (and similarly for $\beta$). In the limit, then, this disc will be both an $\alpha$-stable and a $\beta$-stable structure.

\begin{lem}[Neighbourhood/disc intersections] \label{disc1}  Suppose $\tau'<0.5$. There exists $r$ such that, for $0\ll w \ll n$, all nodes $\boldsymbol{u}$ in $C^{\dagger}_{\boldsymbol{u}_0,rw}$ satisfy  $|\mathcal{N}(\boldsymbol{u}) \cap C^{\dagger}_{\boldsymbol{u}_0,rw}| >\tau'(2w+1)^2$. 
\end{lem}
\begin{proof}
Given $\boldsymbol{u}\in C^{\dagger}_{\boldsymbol{u}_0,rw}$, let $A$ be the set of all nodes in $\mathcal{N}(\boldsymbol{u}) \cap C^{\dagger}_{\boldsymbol{u}_0,rw}$. Then the boundary of $A^{\dagger}$ is a lattice polygon.
By Pick's theorem $\mu(A^{\dagger})=i+\frac{b}{2}-1,$ where $\mu$ denotes Lebesgue measure, $i$ is the number of nodes in the interior of $A^{\dagger}$ and $b$ is the number of nodes on the boundary. For sufficiently large $r$ and for $0 \ll w \ll n$, $\mu(A^{\dagger})>\tau' (2w+1)^2$, so the result follows directly from Pick's theorem. \end{proof}  

For the remainder of this paragraph, fix $r$ as guaranteed by Lemma \ref{disc1} when $\tau'=2\tau_1$. Let $\boldsymbol{u}_0=(x,y)$ and for $\boldsymbol{u}\in C^{\dagger}_{\boldsymbol{u}_0,rw}$ let $\mathcal{N}_1(\boldsymbol{u})=\mathcal{N}(\boldsymbol{u}) \cap C^{\dagger}_{\boldsymbol{u}_0,rw}$.
Now, to establish the even distribution of nodes of each type within $C^{\dagger}_{\boldsymbol{u}_0,rw}$, we divide this region up into a fixed number (i.e.\ independent of $w$ and $n$) of \emph{small neighbourhoods}.  So, for some  $k\in \mathbb{N}^+$, let $w'= \lceil w/k \rceil$, and  consider small neighbourhoods of the form $[x +aw', x+(a+1)w') \times [y +bw', y+(b+1)w')$, where $a,b\in \mathbb{Z}$ and $n\gg aw', n\gg bw'$.  Let $\Pi$ be the set of these small neighbourhoods which lie entirely within $C^{\dagger}_{\boldsymbol{u}_0,rw}$. Given $\boldsymbol{u}\in C^{\dagger}_{\boldsymbol{u}_0,rw}$, let $\mathcal{N}_2(\boldsymbol{u})$ be the set of nodes in $\mathcal{N}_1(\boldsymbol{u})$ which belong to small neighbourhoods in $\Pi$ entirely contained in $\mathcal{N}^{\dagger}(\boldsymbol{u})$. 
For sufficiently large $k$, and for $0\ll w\ll n$, we have that $|\mathcal{N}_2(\boldsymbol{u})|>2\tau_2(2w+1)^2$ for all $\boldsymbol{u}\in C^{\dagger}_{\boldsymbol{u}_0,rw}$. So fix $k$ satisfying this condition.  Finally, choose $\epsilon>0$ such that $(1-\epsilon) \tau_2> \tau_3$. Then, since $k$ is fixed, it follows by the weak law of large numbers that in the limit, the proportion of the nodes in each small neighbourhood in $\Pi$ which are of type $\alpha$ is greater than $0.5(1-\epsilon)$ (and similarly for $\beta$). It then follows that, in the limit, all nodes $u\in   C^{\dagger}_{\boldsymbol{u}_0,rw}$ satisfy the condition that the number of $\alpha$ nodes in $\mathcal{N}_2(\boldsymbol{u})$ is greater than $\tau_3(2w+1)^2$, and so is greater than $\tau_{\alpha}(2w+1)^2$ (and similarly for $\beta$). Thus $C^{\dagger}_{\boldsymbol{u}_0,rw}$ will be both an $\alpha$-stable and a $\beta$-stable structure, as required.

\subsubsection*{The proof of clause (b) of Theorem \ref{main2d}} We suppose we are given a node $\boldsymbol{u}_0=(x_0,y_0)$, chosen uniformly at random. Recalling the discussion of Section \ref{overview},  our first aim is to establish a region $\mathcal{Q}$ containing $\boldsymbol{u}_0$,  of size which means that in the initial configuration (for $0\ll w \ll n$) we shall very likely find $\boldsymbol{u}\in \mathcal{Q}$ such that $\mathtt{ruh}^\alpha_{\tau_{\alpha}}(\boldsymbol{u})$ holds, while at the same time it is very unlikely that we will find any $\boldsymbol{u}$ such that $\mathtt{pn}^{\alpha}_{\gamma,\tau_{\alpha}}(\boldsymbol{u})$ or $ \mathtt{uh}^{\beta}_{\tau_{\beta}}(\boldsymbol{u})$ hold.
The following lemma is what we need in order to do this, and provides the motivation behind our definitions of $\kappa$ and $\tau_{\alpha} \lhd \tau_{\beta}$.  
The proof  appears in Section \ref{defer}, and simply consists of applying standard bounds for the tail of the binomial distribution together with multiple applications of Stirling's approximation. 
 
\begin{lem}[Properties of $\lhd$ and $\kappa$]  \label{zeta} 
Suppose that $\kappa<\tau_{\alpha}<0.5$ and $\tau_{\beta} \lhd \tau_{\alpha}$. Choose $\tau$ such that $\tau_{\beta}\lhd \tau<\tau_{\alpha}$ and $\kappa <\tau$. Then there exists $\gamma>\frac{1}{2}$ and $\zeta>1$ such that for $0\ll w\ll n$, and for a node $\boldsymbol{u}$ selected uniformly at random, $\textbf{P}(\mathtt{ruh}^\alpha_{\tau}(\boldsymbol{u}))>\zeta^w \cdot \textbf{P}(\mathtt{pn}^{\alpha}_{\gamma,\tau_{\alpha}}(\boldsymbol{u}))$ and   $\textbf{P}(\mathtt{ruh}^\alpha_{\tau}(\boldsymbol{u}))>\zeta^w \cdot  \textbf{P}(\mathtt{uh}^{\beta}_{\tau_{\beta}}(\boldsymbol{u}))$.\footnote{Here it is to be understood that how large one must take $w$, and how large $n$ must be compared to $w$, may depend on our particular choices of $\tau$ and $\gamma$ (as well as the given values $\tau_{\alpha}$ and $\tau_{\beta}$).} 
\end{lem} 

For the remainder of this section we assume that  $\tau_{\alpha},\tau_{\beta}$ satisfying the conditions of Lemma \ref{zeta} are fixed. Once and for all we choose (any) $\tau$ such that $\tau_{\beta}\lhd \tau<\tau_{\alpha}$ and $\kappa <\tau$, and then we let $\gamma$ and $\zeta$ be as guaranteed by Lemma \ref{zeta}.  Given this choice of $\gamma$, we then let $r_{\ast}$ be sufficiently large that conditions $(\dagger_a)$ and $(\dagger_b)$ of Section \ref{overview} are satisfied. We also assume we are given $\epsilon_0>0$. Our  aim is to show that for $0 \ll w \ll n$, the probability $\boldsymbol{u}_0$ will be of type $\beta$ in the final configuration is $>1-\epsilon_0$.  The basic approach is to establish that in the limit $\boldsymbol{u}_0$ will be consumed by a $\beta$-firewall which originates with $\mathcal{Q}$. 

\subsubsection*{Defining $\mathcal{Q}$} To specify $\mathcal{Q}$  it will be useful to consider a grid of nodes with disjoint neighbourhoods, meaning that these neighbourhoods are independently distributed in the initial configuration: 
 
 \begin{defin}[Grid] 
 Recall that $\boldsymbol{u}_0=(x_0,y_0)$. We say that $\boldsymbol{v}=(x,y)$ is \emph{on the grid} if there exist $a,b\in \mathbb{Z}$ with $|a|,|b|<0.25n/(2w+1)$ and $x=x_0 +a(2w+1)$, $y=y_0 +b(2w+1)$. Suppose $m=2k+1$.  Then by the $m$-\emph{square on the grid centred at} $\boldsymbol{u}_0$ we mean the set of nodes $ \{ (x_0+a(2w+1), y_0+b(2w+1)):\ a,b\in [-k, +k] \}$.  
 \end{defin}

\noindent Now consider the initial configuration and  let $\pi(w)\in \mathbb{R}$ be such that for $\boldsymbol{u}$ chosen uniformly at random, $\textbf{P}(\mathtt{ruh}^\alpha_{\tau}(\boldsymbol{u}))=1/\pi(w)$.  Note that  $1/\pi(w)\rightarrow 0$ as $w\rightarrow \infty$, and recall that  $\left( 1- \frac{1}{x} \right)^{x} \rightarrow \frac{1}{e}$ as $x\rightarrow \infty$. Take $k_0$ such that  $e^{-k_0}\ll \epsilon_0$ (with $\epsilon_0$ as fixed previously), noting that $k_0$ does not depend on $w$. Take the least  odd number $m_0(w)\in \mathbb{N}^+$ such that $m_0(w)^2 \geq \pi(w)k_0$.    Let $\mathcal{Q}_0$ be the  $m_0(w)$-square on the grid centred at $\boldsymbol{u}_0$. Then the probability that not a single node $\boldsymbol{u}\in \mathcal{Q}_0$ satisfies $\mathtt{ruh}^\alpha_{\tau}(\boldsymbol{u})$ is at most: $$\left( 1- \frac{1}{\pi(w)} \right)^{\pi(w)k_0}$$ which is $ \ll \epsilon_0 $  for sufficiently large $w$. So far then, we have identified a set of nodes $\mathcal{Q}_0$, which has the property that the following will fail to be true with probability $\ll \epsilon_0$: there exists $\boldsymbol{u}\in \mathcal{Q}_0$ such that  $\mathtt{ruh}^{\alpha}_{\tau}(\boldsymbol{u})$ holds.

We need a little more from the vicinity of $\boldsymbol{u}_0$  that we are going to work with. In order to ensure that firewalls created inside $\mathcal{Q}^{\dagger}_0$ will grow to include $\boldsymbol{u}_0$, we need a much larger region in which we will very likely not have any nodes $\boldsymbol{u}$ for which    
either of $\mathtt{pn}_{\gamma,\tau_{\alpha}}^{\alpha}(\boldsymbol{u})$ or $\mathtt{uh}^{\beta}_{\tau_{\beta}}(\boldsymbol{u})$ hold. 
We let $\mathcal{Q}_1$ be the $(8w+1)m_0(w)$-square on the grid centred at $\boldsymbol{u}_0$. Then we define $\mathcal{Q} = \mathcal{Q}_1^{\dagger}$.  Let  $p$ be the probability  that any node $\boldsymbol{u}\in \mathcal{Q}$ satisfies $\mathtt{pn}^{\alpha}_{\gamma, \tau_{\alpha}}(\boldsymbol{u})$. Since the number of nodes in $\mathcal{Q}$ is $\leq \pi(w)k_0 (8w+1)^4$, it follows that $p$ is at most  $\pi(w)k_0 (8w+1)^4$ times the probability that  $\mathtt{pn}^{\alpha}_{\gamma, \tau_{\alpha}}(\boldsymbol{u})$ holds for $\boldsymbol{u}$ selected uniformly at random. So by Lemma \ref{zeta},  $p \ll \epsilon_0$ for sufficiently large $w$.  A similar argument holds for $\mathtt{uh}^{\beta}_{\tau_{\beta}}(\boldsymbol{u})$. Thus we conclude that for $0\ll w \ll n$, $\mathtt{Typical}\   \mathtt{vicinity}$  fails to hold with probability $\ll \epsilon_0$: 

\begin{defin}[Typical vicinity] \label{typ} 
We say that $\mathtt{Typical}\   \mathtt{vicinity}$  holds if all of (1)--(3) below are true: 
\begin{enumerate} 
\item $\mathtt{ruh}^{\alpha}_{\tau}(\boldsymbol{u}_0)$ does not hold; 
\item There exists $\boldsymbol{u}\in \mathcal{Q}_0$ such that $\mathtt{ruh}^{\alpha}_{\tau}(\boldsymbol{u})$ holds; 
\item There does not exist any $\boldsymbol{u}$ in $\mathcal{Q}$ such that either of  $\mathtt{pn}_{\gamma,\tau_{\alpha}}^{\alpha}(\boldsymbol{u})$ or $\mathtt{uh}^{\beta}_{\tau_{\beta}}(\boldsymbol{u})$ hold. 
\end{enumerate} 
\end{defin} 

Note that (1)--(3) from Definition \ref{typ} all refer to the initial configuration. 

\begin{defin} Consider the initial configuration. 
For $\tau'\in [0,1]$ we say that $\mathtt{ju}^{\alpha}_{\tau'}(\boldsymbol{u})$ holds if there are less than $\tau'(2w+1)^2$ many $\alpha$ nodes in $\mathcal{N}(\boldsymbol{u})$, but changing the type of 2w+1 $\beta$ nodes in $\mathcal{N}(\boldsymbol{u})$ would cause this not to be the case. We say that $\mathtt{rju}^{\alpha}_{\tau'}(\boldsymbol{u})$ holds if there are less than $\tau'(2w+1)(3w+1)$ many $\alpha$ nodes in $\mathcal{N}^{\diamond}(\boldsymbol{u})$, but  changing the type of 3w+1 $\beta$ nodes in $\mathcal{N}^{\diamond}(\boldsymbol{u})$ would cause this not to be the case.\footnote{Think of $\mathtt{ju}$ as j-ust u-nhappy and think of $\mathtt{rju}$ as r-ight extended neighbourhood j-ust u-nhappy.}    
\end{defin} 

Observe that if $\mathtt{Typical}\   \mathtt{vicinity}$  holds, then we are guaranteed the existence of $\boldsymbol{u}\in \mathcal{Q}_0^{\dagger}$ for which $\mathtt{rju}^{\alpha}_{\tau}(\boldsymbol{u})$ holds. 

\subsubsection*{Smoothness conditions} Our aim will be to show that for $0\ll w \ll n$, a large $\beta$-firewall can be expected to form in the early stages of the process in the vicinity of  $\boldsymbol{u}\in \mathcal{Q}_0^{\dagger}$ for which $\mathtt{rju}^{\alpha}_{\tau}(\boldsymbol{u})$ holds. To this end, however, we must first examine what can be expected in the initial configuration, from the vicinity of $\boldsymbol{u}$ which is chosen uniformly at random from amongst the nodes such that $\mathtt{rju}^{\alpha}_{\tau}(\boldsymbol{u})$ holds. In particular we are interested in the three regions:  $C^0(\boldsymbol{u}):=C_{\boldsymbol{u},r_\ast w}^{\dagger}$, $C^1(\boldsymbol{u}):= C_{\boldsymbol{u},2r_\ast w}^{\dagger}$ and $C^2(\boldsymbol{u}):=C_{\boldsymbol{u},3r_\ast w}^{\dagger}$. In this subsection, we look to establish certain smoothness conditions -- that the types of nodes in these regions will be nicely distributed.

 For some large $k_1$, which we shall specify later and which will not depend on $w$, the basic idea now is that we want to cover the nodes in $C^2(\boldsymbol{u})$ with disjoint $\frac{w}{k_1}$-squares. This occasions the minor inconvenience that $k_1$ may not divide $w$. We therefore let $\mathcal{I}_{k_1}(\boldsymbol{u})$ be a pairwise disjoint set of rectangles, whose union contains  all nodes in $C^2(\boldsymbol{u})$, and such that: 
\begin{itemize} 
\item  Each element of $\mathcal{I}_{k_1}(\boldsymbol{u})$ has nonempty intersection with $C^2(\boldsymbol{u})$ and is of the form $[x, x+a) \times [y, y+b)$ for some 
$a,b\in \{ \lfloor w/k_1 \rfloor, \lceil w/k_1 \rceil, \lceil w/k_1 \rceil +1\}$ and  $x,y\in \mathbb{N}$;
\item Each element of $\mathcal{I}_{k_1}(\boldsymbol{u})$ is either entirely contained in $\mathcal{N}^{\diamond}(\boldsymbol{u})$, or else is disjoint from $\mathcal{N}^{\diamond}(\boldsymbol{u})$. 
\end{itemize}

\begin{defin}[Smoothness for $\boldsymbol{u}$] \label{smoothdef} 
Suppose given $k_1\in \mathbb{N}^+$ and $\epsilon_1>0$ and $\boldsymbol{u}$ such that $\mathtt{rju}^{\alpha}_{\tau}(\boldsymbol{u})$ holds.   We say that $\mathtt{Smooth}_{k_1,\epsilon_1}(\boldsymbol{u})$ holds if: 
\begin{itemize} 
\item For each $A\in \mathcal{I}_{k_1}(\boldsymbol{u})$ which is contained in $\mathcal{N}^{\diamond}(\boldsymbol{u})$, the proportion of the nodes in $A$ which are of type $\alpha$ is in the interval $[\tau -\epsilon_1, \tau+\epsilon_1]$. 
\item For each $A\in \mathcal{I}_{k_1}(\boldsymbol{u})$ which is disjoint from  $\mathcal{N}^{\diamond}(\boldsymbol{u})$, the proportion of the nodes in $A$ which are of type $\alpha$ is in the interval $[0.5 -\epsilon_1, 0.5+\epsilon_1]$. 
\end{itemize} 
\end{defin}

The following lemma is then almost immediate and is proved in Section \ref{defer}: 

\begin{lem}[Likely smoothness]  \label{smoothlem} 
Suppose given $k_1\in \mathcal{N}^+$ and $\epsilon_1, \epsilon>0$. Suppose that $\boldsymbol{u}$ is selected uniformly at random from amongst the nodes such that $\mathtt{rju}^{\alpha}_{\tau}(\boldsymbol{u})$ holds. For $0\ll w\ll n$, the probability that  $\mathtt{Smooth}_{k_1,\epsilon_1}(\boldsymbol{u})$ holds is greater than $1-\epsilon$. 
\end{lem} 

\begin{defin}[The smoothness event] 
We let $\boldsymbol{u}_1$ be chosen uniformly at random from amongst the nodes $\boldsymbol{u}\in \mathcal{Q}_0^{\dagger}$ such that 
$\mathtt{rju}^{\alpha}_{\tau}(\boldsymbol{u})$ holds (so that if there exists no such node then $\boldsymbol{u}_1$ is undefined). For any $k_1, \epsilon_1>0$, we let $\mathtt{Smooth}(k_1,\epsilon_1)$ be the event that $\boldsymbol{u}_1$ is defined and  
$\mathtt{Smooth}_{k_1,\epsilon_1}(\boldsymbol{u_1})$ holds. 
\end{defin} 

\noindent Our arguments so far suffice to show that, for any $k_1, \epsilon_1>0$, if $0\ll w \ll n$ then the probability that $\mathtt{Smooth}(k_1,\epsilon_1)$ fails to hold is $\ll \epsilon_0$.  We now want to examine what satisfaction of $\mathtt{Smooth}(k_1,\epsilon_1)$ can tell us about the proportion of $\alpha$ nodes in regions contained in $C^2(\boldsymbol{u}_1)$ but which are not one of the rectangles in $\mathcal{I}_{k_1}(\boldsymbol{u}_1)$.  In order to do so we define a couple of  functions, which describe the proportion of $\alpha$ nodes in a given set $A$, either in the initial configuration or else after all $\alpha$ nodes in $B\subseteq A$ have changed type. We also consider idealised versions of these functions which will be easier to work with most of the time. Here and elsewhere $\bar B$ denotes the complement of $B$.

\begin{defin} \label{counters} 
Consider sets of nodes $A,B\subset C^2(\boldsymbol{u})$. 
\begin{enumerate} 
\item  Let $\Xi(A)$ be the proportion of the elements of $A$ which are of type $\alpha$ in the initial configuration.
\item  Let $A_0=A \cap \mathcal{N}^{\diamond}(\boldsymbol{u})$ and $A_1= A- A_0$. Define $\Xi^{\ast}(A,\boldsymbol{u})=(\tau |A_0| +0.5|A_1|)/|A|$. 
\item Take the initial configuration and then change all nodes in $B$ to type $\beta$. Let  $\Xi(A,B)$ denote  the proportion of the elements of $A$ which are now of type $\alpha$. 
\item Let $A_0=A \cap B, A_1= A\cap \bar B \cap \mathcal{N}^{\diamond}(\boldsymbol{u})$ and let $A_2= A-(A_0 \cup A_1)$. We define $\Xi^{\ast}(A,\boldsymbol{u},B)= (\tau |A_1| +0.5 |A_2|)/|A|$.  
 \end{enumerate} 
\end{defin} 

\noindent So $\Xi^{\ast}(A,\boldsymbol{u})$ gives the proportion of nodes in $A$ which would be of type $\alpha$, if \emph{exactly} proportion $\tau$ of those nodes in $A\cap \mathcal{N}^{\diamond}(\boldsymbol{u})$ were of type $\alpha$, and \emph{exactly} half of the nodes in the remainder of $A$ were of type $\alpha$.  On the other hand $\Xi^{\ast}(A,\boldsymbol{u}, B)$ gives the corresponding proportion if the same conditions hold, but then we change all nodes in $B$ to type $\beta$.  One may think of $\Xi^{\ast}$ as an idealised version of $\Xi$. Satisfaction of $\mathtt{rju}^{\alpha}_{\tau}(\boldsymbol{u})$ and $\mathtt{Smooth}_{k_1,\epsilon_1}(\boldsymbol{u})$  for large $k_1$ and small $\epsilon_1$ will ensure that $\Xi^{\ast}$ is a reasonable approximation to $\Xi$, in a sense that we will make precise. 

Now we want to work with a greater variety of sets of nodes than just those in $\mathcal{I}_{k_1}(\boldsymbol{u})$, but we still only need to consider sets which are reasonably large and of a reasonably simple form:

\begin{defin} Given a node $\boldsymbol{u}$ and $k_2\in \mathcal{N}^+$, we let $\mathcal{I}^{\ast}_{k_2}(\boldsymbol{u})$ be the set of all sets of nodes $A$ such that $A\subseteq C^2(\boldsymbol{u})$ and either (1) $A$ is an $a\times b$ rectangle for $a,b\geq w/k_2$, or (2) the set of nodes inside a regular polygon with sides of length $\geq w/k_2$, or (3) the union of two sets of the form (1) or (2).  
\end{defin}

So long as we restrict attention to sets of nodes in $\mathcal{I}^{\ast}_{k_2}(\boldsymbol{u})$ (for some $k_2$ to be specified), the following observation allows us to work with the idealised functions $\Xi^{\ast}$: 

\begin{obs}[Smoothness for elements of $\mathcal{I}^{\ast}_{k_2}(\boldsymbol{u})$]  \label{ideal okay} Suppose given $k_2\in \mathbb{N}^+$ and $\delta>0$. If $k_1$ is sufficiently large and $\epsilon_1>0$ is sufficiently small then, for $0 \ll w \ll n$, satisfaction of $\mathtt{rju}^{\alpha}_{\tau}(\boldsymbol{u})$ and $\mathtt{Smooth}_{k_1,\epsilon_1}(\boldsymbol{u})$ suffices to ensure that: 
\begin{enumerate}[] 
\item  \hspace{0.3cm}$(\fgeeszett)$ \hspace{0.3cm}For all $A,B\in \mathcal{I}^{\ast}_{k_2}(\boldsymbol{u})$, $|\Xi(A)-\Xi^{\ast}(A,\boldsymbol{u})|<\delta$ and $|\Xi(A,B)-\Xi^{\ast}(A,\boldsymbol{u},B)|<\delta$. 
\end{enumerate}
\end{obs}

\subsubsection*{Establishing the creation of firewalls} Now we choose values of $k_2$ and $\delta$ which will allow us to argue that a large firewall is very probably created around $\boldsymbol{u}_1$ in the early stages of the process. With these values specified, $k_1$ and $\epsilon_1$ are simply chosen to be those values guaranteed by Lemma \ref{ideal okay}, meaning that we can make use of our idealised functions $\Xi^{\ast}$. Numerical values in the following definition are somewhat arbitrary, but suffice for our purposes. 

\begin{defin}[Choosing $k_2,\delta,k_1$ and $\epsilon_1$] \label{k1} Choose $k_2 > (0.5-\tau)/(\tau_{\alpha}-\tau)$, choose  $\delta>0$ such that $\delta \ll  \mbox{min}\{ (\tau_{\alpha}-\tau)(2\tau-0.5)/(2k_2), 10^{-5} \}$ and choose $k_1$ sufficiently large and $\epsilon_1>0$ sufficiently small that for $0 \ll w \ll n$ satisfaction of $\mathtt{rju}^{\alpha}_{\tau}(\boldsymbol{u})$ and $\mathtt{Smooth}_{k_1,\epsilon_1}(\boldsymbol{u})$ suffices to ensure satisfaction of $(\fgeeszett)$ (as specified in Observation \ref{ideal okay}). 
\end{defin}

The next lemma finally establishes that a firewall of radius $r_{\ast}w$ very probably forms around $\boldsymbol{u}_1$. 

\begin{lem}[Firewalls] \label{firewall}
Suppose that $\mathtt{rju}^{\alpha}_{\tau}(\boldsymbol{u})$ and $\mathtt{Smooth}_{k_1,\epsilon_1}(\boldsymbol{u})$ both hold. Let $t^{\ast}=2r_\ast w + 3w$ and suppose that there are no hopeful $\beta$ nodes in $C^2(\boldsymbol{u})$ at any stage $\leq t^{\ast}$. Then  all nodes in $C^0(\boldsymbol{u})$ are of type $\beta$ at stage $t^{\ast}$. 
\end{lem}
\begin{proof} 
In what follows it will be convenient to assume that $w$ is even. Only small modifications are required to deal with the case that $w$ is odd. It is also convenient to assume that $r_\ast >2$. The proof will basically consist of repeated applications of condition $(\fgeeszett)$. We shall apply this condition in order to inductively establish a  sequence of increasingly large rectangles for which all nodes become of type $\beta$.  Let $\boldsymbol{u}=(x,y)$, let $x^\ast:=x+\frac{w}{2}$  and  $\boldsymbol{u}^\ast:=(x^\ast,y)$. The following argument is illustrated in the first picture of Figure \ref{oct}. 
The rectangles we consider are as follows (where $a\in \mathbb{N}^+$): 

\begin{itemize} 
\item We let $R^0_a$ be the set of nodes $\boldsymbol{v}=(x',y')$ such that $0\leq x'-x\leq w$ and $|y-y'|\leq a$;
\item We let $R^1_a$ be the set of nodes $\boldsymbol{v}=(x',y')$ such that $x-a\leq x' \leq x+w+a$ and $|y-y'|\leq w$; 
\item We let $R^2_a$ be the set of nodes $\boldsymbol{v}=(x',y')$ such that $|x'-x^{\ast}|\leq \frac{w}{2}+ \lceil \frac{\sqrt{2}-1}{2}w \rceil$ and $|y'-y|\leq w+a$. 
\end{itemize} 

\begin{figure}
 \centering
 \scalebox{0.5}{
 \begin{tikzpicture}
\filldraw[fill=gray!20] (-1.5,0) -- (9.5,0) -- (9.5,8) --  (-1.5,8) -- (-1.5,0);
\filldraw[fill=purple!20] (0,0) -- (8,0) -- (8,8) --  (0,8) -- (0,0);
\filldraw[fill=red!20]  (1.5,-0.8) -- (6.5,-0.8) -- (6.5,8.8) --  (1.5,8.8) -- (1.5,-0.8);
\filldraw[fill=green!20] (2.2,3) -- (5.8,3) -- (5.8,5) --  (2.2,5) -- (2.2,3);
\draw[color=black] (-1.5,0) -- (9.5,0) -- (9.5,8) --  (-1.5,8) -- (-1.5,0);
\draw[very thick,color=purple] (0,0) -- (8,0) -- (8,8) --  (0,8) -- (0,0);
\draw[very thick,color=red] (1.5,-0.8) -- (6.5,-0.8) -- (6.5,8.8) --  (1.5,8.8) -- (1.5,-0.8);
\draw[thick,color=black] (2.2,0) -- (5.8,0) -- (5.8,8) --  (2.2,8) -- (2.2,0);
\draw[very thick,color=black] (2.2,3) -- (5.8,3) -- (5.8,5) --  (2.2,5) -- (2.2,3);
\node[fill=gray,circle] (b) at (4,4) {};
\node[fill=red,circle] (e) at (2.2,4) {};
\node[] (b) at (4.6,4) {{\Large $\mathbf{u}^{\ast}$}};
\node[] (k) at  (2.6,4) {{\Large $\mathbf{u}$}};
\node[] (d) at (-0.5,3) {{\Large $R^1_b$}};
\node[] (a) at (5,5.5) {{\Large $R^0_\alpha$}};
\node[] (c) at (9,-0.5) {{\Large $\mathcal{N}^{\diamond}(\mathbf{u})$}};
\draw[<->] (0,7)--(2.2,7);
\node[] (s) at (0.5,7.3) {{\Large $b$}};
\draw[<->] (2.6,8)--(2.6,8.8);
\draw[<->] (3,4)--(3,3);
\node[] (sss) at (3.2,3.5) {{\Large $a$}};
\node[] (stt) at (2.8,8.4) {{\Large $c$}};
\node[] (a) at (1.1,9) {{\Large $R^2_c$}};
\end{tikzpicture}\hspace{0.9cm}
\begin{tikzpicture}
 \filldraw[fill=green!20,draw=green!50!black] (0,5) -- (3,8) -- (5,8) -- (8,5) -- (8,3) -- (5,0) -- (3,0) -- (0,3) -- (0,5);
 \draw[color=black] (0,0) -- (8,0) -- (8,8) --  (0,8) -- (0,0);
\node[fill=gray,circle] (b) at (4,4) {};
\node[] (b) at (4.6,4) {{\Large $\mathbf{u}^{\ast}$}};
\node[] (a) at (6,2) {{\Large $O_r(\mathbf{u}^{\ast})$}};
\node[] (c) at (0,-0.5) {{\Large $S_r(\mathbf{u}^{\ast})$}};
\end{tikzpicture}
} 
\caption{Rectangles in the proof of Lemma \ref{firewall} and the octagon of Definition \ref{definition:octa}}
\label{oct}
\end{figure}
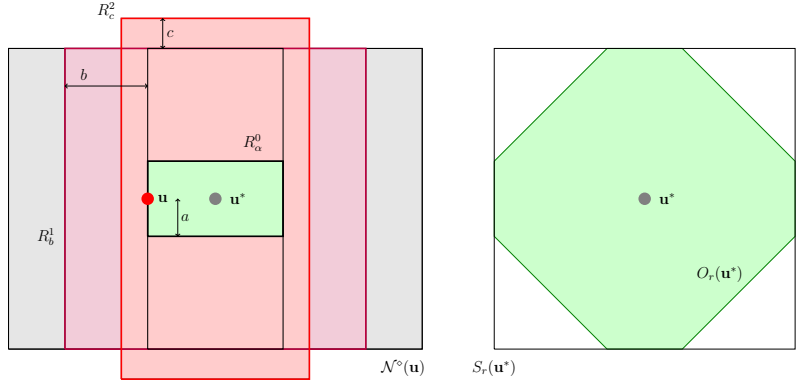

Working always for sufficiently large $w$, we show first that all nodes in $R^0_{\lfloor w/k_2 \rfloor}$ are of type $\beta$ by the end of stage 1.  Then we show inductively that: 
\begin{enumerate}[] 
\item\hspace{0.3cm}$(\dagger_1)$\hspace{0.3cm}All nodes in  $R^0_a$ are of type $\beta$ by the end of stage $a$, for $a\leq w/2$;
\item\hspace{0.3cm}$(\dagger_2)$\hspace{0.3cm}The same result holds for $R^0_a$ when $a\leq w$; 
\item\hspace{0.3cm}$(\dagger_3)$\hspace{0.3cm}All nodes in $R^1_a$ 
are of type $\beta$ by the end of stage $w+a$, for $a\leq \lceil \frac{w}{\sqrt{2}} \rceil$; 
\item\hspace{0.3cm}$(\dagger_4)$\hspace{0.3cm}All nodes in $R^2_a$ 
are of type $\beta$ by the end of stage $2w+a$, for all $a\leq \lceil \frac{\sqrt{2}-1}{2}w \rceil$. 
\end{enumerate}
From there we shall be able to argue that all nodes in $C^0(\boldsymbol{u})$ will eventually be of type $\beta$. 

So our first task is to establish that all nodes in $R^0_{\lfloor w/k_2 \rfloor}$ will be of type $\beta$ by the end of stage 1.  We have that  $R^0_{\lfloor w/k_2 \rfloor}\in I^{\ast}_{k_2}(\boldsymbol{u})$ and  by our choice of $k_2$ the following holds for  all nodes $\boldsymbol{v}$ in this rectangle: 
\[ \Xi^{\ast}(\mathcal{N}(\boldsymbol{v})) \leq   \frac{\tau + \tau_{\alpha}}{2} . \] 
Since $\delta< (\tau_{\alpha}-\tau)/2$, it follows from $(\fgeeszett)$ that for $0\ll w \ll n$, $\Xi(\mathcal{N}(\boldsymbol{v})) < \tau_{\alpha}$. So for $0\ll w \ll n$,   all $\alpha$ nodes in $R^0_{\lfloor w/k_2 \rfloor}$ are hopeful in the initial configuration and will be of type $\beta$ by the end of stage 1, as required.

For any $a\in \mathbb{N}$, let $d_a=a/(2w+1)$. In order to inductively establish $(\dagger_1)$ -- $ (\dagger_4)$ we then need the various technical facts contained in the following proposition, which is easily verified by direct calculation. 

\begin{prop} \label{tech} Given $\epsilon >0$, the following all hold for $0 \ll w \ll n$: 
\begin{enumerate} 
\item For all $a\in (\lfloor w/k_2 \rfloor,  w/2]$ and all $\boldsymbol{v}\in R^0_a-R^0_{a-1}$, 
 $ |\Xi^{\ast}(\mathcal{N}(\boldsymbol{v}),\boldsymbol{u},R^0_{a-1}) -( \tau + (0.5-2\tau)d_a)| <\epsilon. $
 \item For all $a\in  (w/2, w]$ and $\boldsymbol{v}\in R^0_a - R^0_{a-1}$, $\Xi^{\ast}(\mathcal{N}(\boldsymbol{v}), R^0_{a-1}) < 0.25 +0.25\tau +\epsilon$. 
 \item For all $a\in  (0, \lceil \frac{w}{\sqrt{2}} \rceil]$ and $\boldsymbol{v}\in R^1_a - R^1_{a-1}$, $\Xi^{\ast}(\mathcal{N}(\boldsymbol{v}), R^1_{a-1}) < 3/8- (1/2- \tau)(1-1/\sqrt{2})/4 +\epsilon$.
 \item For all $a\in (0,\lceil w(\sqrt{2}-1)/2 \rceil]$ and $\boldsymbol{v}\in R^2_a - R^2_{a-1}$, $\Xi^{\ast}(\mathcal{N}(\boldsymbol{v}), R^2_{a-1} \cup R^1_{\lceil \frac{w}{\sqrt{2}} \rceil}) <3/8- (1/16)(3/2-1/\sqrt{2}) +\epsilon$. 
\end{enumerate} 

\end{prop}

Given (1) of Proposition \ref{tech}, the induction to establish $(\dagger_1)$ now goes through easily. Applying $(\fgeeszett)$ we have that if $0\ll w \ll n$ then for $a\in (\lfloor w/k_2 \rfloor,  w/2]$ and $\boldsymbol{v}\in R^0_a-R^0_{a-1}$, $\Xi(\mathcal{N}(\boldsymbol{v}),R_{a-1})\leq \tau+ (0.5-2\tau)d_a +2\delta$. Since $d_a >1/3k_2$ and $\delta<(2\tau-0.5)/(6k_2)$ the induction step follows. 
  
 Similarly (2) of Proposition \ref{tech} gives us the induction step in establishing $(\dagger_2)$. Applying $(\fgeeszett)$ we have that if $0\ll w \ll n$ then for $a\in  (w/2, w]$ and $\boldsymbol{v}\in R^0_a - R^0_{a-1}$, $\Xi(\mathcal{N}(\boldsymbol{v}), R^0_{a-1}) < 0.25 +0.25\tau +2\delta$. Now  for $\tau\in (\kappa,0.5)$, $0.25 +0.25\tau <\tau -0.02$ so, since  $\delta <10^{-5}$, the induction step follows.

 For $(\dagger_3)$ we have that  if $0\ll w \ll n$ then for $a\in  (0, \lceil \frac{w}{\sqrt{2}} \rceil]$ and $\boldsymbol{v}\in R^1_a - R^1_{a-1}$, $\Xi(\mathcal{N}(\boldsymbol{v}), R^0_{a-1}) < 3/8- (1/2- \tau)(1-1/\sqrt{2})/4 +2\delta$. 
 Now for $\tau\in (\kappa, 0.5)$ we have $\tau- (3/8- (1/2- \tau)(1-1/\sqrt{2})/4)> 2\cdot 10^{-5}>2\delta$, so the induction step follows.

  Finally, for $(\dagger_4)$ we have that  if $0\ll w \ll n$ then for $a\in (0,\lceil w(\sqrt{2}-1)/2 \rceil]$ and $\boldsymbol{v}\in R^2_a - R^2_{a-1}$, $\Xi(\mathcal{N}(\boldsymbol{v}), R^2_{a-1}\cup R^1_{\lceil \frac{w}{\sqrt{2}} \rceil}) <3/8- 1/16(3/2-1/\sqrt{2}) +2\delta$. Then  for $\tau\in (\kappa, 0.5)$ we have 
  \[
  \tau - (3/8- (1/16)(3/2-1/\sqrt{2})) >0.02 >2\delta.
  \]
   Once again the induction step goes through. 
  
  So far we have established that all nodes in $R:= R^1_{ \lceil w/\sqrt{2} \rceil} \cup R^2_{\lceil w(\sqrt{2}-1)/2 \rceil}$ will be of type $\beta$ by the end of stage $3w$. Now we wish to extend this and argue that all nodes in $C^0(\boldsymbol{u})$ will eventually be of type $\beta$. Previously we observed that, for $0\ll w \ll n$, any $\alpha$ node $\boldsymbol{v}$ on the outer boundary of $C^0(\boldsymbol{u})$ will be hopeful, so long as  $\mathtt{pn}^{\alpha}_{\gamma,\tau_{\alpha}}(\boldsymbol{v})$ does not hold. In order to get to the point where we can conclude that all nodes in $ C^0(\boldsymbol{u})$ will become of type $\beta$, we use a similar idea, but we must work with smaller regions than $C^0(\boldsymbol{u})$ -- this will not be a problem, since satisfaction of $\mathtt{Smooth}_{k_1,\epsilon_1}(\boldsymbol{u})$ guarantees the failure of much weaker conditions than $\mathtt{pn}^{\alpha}_{\gamma,\tau_{\alpha}}(\boldsymbol{v})$ for $\boldsymbol{v}\in C^1(\boldsymbol{u})$. In fact the calculations will be simpler if we work with regular octagons rather than circles: 
  
  \begin{defin}[Octagon] \label{definition:octa}Given $r\in \mathbb{R}$, we let $S_{r}(\boldsymbol{u}^{\ast})$ be the square in $\mathbb{R}^2$ centred at $\boldsymbol{u}^{\ast}$ with sides of length $r$ parallel to the axes. Then we let $O_{r}(\boldsymbol{u}^{\ast})$ be the largest regular octagon contained in $S^{\dagger}_{r}(\boldsymbol{u}^{\ast})$, as illustrated in the second image of Figure \ref{oct}. 
  \end{defin} 
Note that for $r_0= w(1+\sqrt{2})$, $O_{r_0}(\boldsymbol{u}^{\ast})$ has sides of length $w$ and  that $O^{\dagger}_{r_0}(\boldsymbol{u}^{\ast})$ is entirely contained in $R^{\dagger}$ (where $R$ is as specified above). Now if $r\geq r_0$ then for $0 \ll w \ll n$, if $\boldsymbol{v}\in C^{1}(\boldsymbol{u})$ is an $\alpha$ node on the outer boundary of $O_{r}(\boldsymbol{u}^{\ast})$ then $\Xi(\mathcal{N}(\boldsymbol{v})) <11/32 +2 \delta<\tau$, and so $\boldsymbol{v}$ is hopeful.  Inductively we  conclude that all nodes in $C^1(\boldsymbol{u}) \cap O_{r_0+2a}^{\dagger}(\boldsymbol{u}^{\ast})$ are of type $\beta$ by the end of stage $3w+a$. Now there exists $a$ with $C^0(\boldsymbol{u}) \subset O^{\dagger}_{r_0+2a}(\boldsymbol{u}^{\ast}) \subset C^1(\boldsymbol{u})$. Since $a<2r_\ast w$, the result follows as required.  
\end{proof} 
 
The following lemma completes our proof of clause (b) of Theorem \ref{main2d}.  

\begin{lem}[Eventual conversion] \label{final} 
Suppose that $\mathtt{Typical}\   \mathtt{vicinity}$  and $\mathtt{Smooth}(k_1,\epsilon_1)$ both hold. Then there exists a stage $s$ after which $\boldsymbol{u}_0$ is always of type $\beta$. 
\end{lem}
\begin{proof} 
Recall that we defined $\mathcal{Q}_0$ to be the $m_0(w)$-square on the grid centred at $\boldsymbol{u}_0$, and that we defined  $\mathcal{Q}_1$ to be the $(8w+1)m_0(w)$-square on the grid centred at $\boldsymbol{u}_0$. Now we define  $\mathcal{Q}_2$ to be the $3m_0(w)$-square on the grid centred at $\boldsymbol{u}_0$. First of all we observe that we are guaranteed a large number of stages before the existence of any unhappy $\beta$ nodes in $\mathcal{Q}_2$. This follows because a $\beta$ node which is  happy in the initial configuration, cannot become unhappy until strictly after the first stage (if such a stage exists) at which another  $\beta$ node in its neighbourhood becomes unhappy. For $0\ll w \ll n$ we have that in the initial configuration, if $m= \mbox{min} \{ || \boldsymbol{u}- \boldsymbol{v} ||_{\infty}: \ \boldsymbol{v}\in \mathcal{Q}_2,\  \boldsymbol{u} \ \mbox{ is an unhappy } \beta \mbox{ node} \}$, then $m>3wm_0(w)(2w+1)$. Thus we are guaranteed not to find any unhappy $\beta$ nodes in $\mathcal{Q}_2$ at stages prior to $t^{\ast}= 3m_0(w)(2w+1)$. 

Applying Lemma \ref{firewall} we conclude that for $0\ll w \ll n$ we shall have that by stage $2r_\ast w + 3w$ all nodes in $C^0(\boldsymbol{u}_1)$ will be of type $\beta$. For $0\ll w \ll n$ it then follows inductively, by the choice of $r_\ast$,  that at each stage $<2m_0(w) (2w+1)$ and after the creation of this firewall, its radius will expand by at least 1. Since $|\boldsymbol{u}_0-\boldsymbol{u}_1|<m_0(w)(2w+1)$ we conclude that there exists some stage after which $\boldsymbol{u}_0$ always belongs to a $\beta$-firewall of radius $>r_\ast w$ centred at $\boldsymbol{u}_1$. 
\end{proof}

\subsubsection*{The proof of clause (c) of Theorem \ref{main2d}} This case is much simpler. Recall that in the proof of clause (b) of Theorem \ref{main2d} we chose $r_{\ast}$ so that certain conditions were satisfied.  Now that we have $\tau_{\beta}<0.5 <\tau_{\alpha}$, however, we can simply choose $r_{\ast}$ large enough such that for $0\ll w \ll n$ if $r\geq r_{\ast}$ then for any $\beta$ node $\boldsymbol{v}$ inside a $\beta$ firewall  of radius $rw$, close enough to a half of $\mathcal{N}(\boldsymbol{v})$ will lie inside the firewall to ensure that $\boldsymbol{v}$ is happy, and similarly any $\alpha$ node $\boldsymbol{v'}$ on the outer boundary of such a firewall  will be hopeful.  Such a firewall must then spread until every node is eventually contained inside it. For any $\epsilon>0$, if $n$ is sufficiently large then there will exist such a firewall in the initial configuration with probability $>1-\epsilon$. Whenever there exists such a firewall all nodes must eventually be of type $\beta$.

\section{The proofs of Theorems \ref{main2d} (d) and \ref{main2d} (e)} \label{upper half} 

In this section we work with $\tau_{\alpha},\tau_{\beta}>0.5$. The proofs are simple modifications of the proofs of Theorems \ref{main2d}(b) and \ref{main2d}(a). Rather than describing those proofs again in their entirety, with only small changes, we describe the necessary modifications. Roughly speaking, the idea is that we now replace considerations as to whether a $\beta$ node $\boldsymbol{u}$ has less than proportion $\tau_{\beta}$ many $\beta$ nodes in its neighbourhood, with the question as to whether it has proportion $\leq 1-\tau{_\alpha}$ many $\beta$ nodes. If this holds then $\boldsymbol{u}$ is unhappy as type $\beta$, but would be happy if it changed type (in fact we should take into account the effect of the type change of $\boldsymbol{u}$ on the proportion, meaning that $1-\tau_{\alpha}$ is not exactly the right proportion to consider but approaches it for large $w$). The nodes of this type, then, will be the hopeful nodes -- and these are the nodes which may initiate firewalls. Similarly, we replace the notion of an $\alpha$-stable structure, with that of an $\alpha$-intractable structure: 

\begin{defin}[Intractable structures] 
We say that a set of nodes $A$ is an $\alpha$-intractable structure if it contains $\beta$ nodes, and for every $\beta$ node $\boldsymbol{u}\in A$, $\boldsymbol{u}$ would be unhappy (as an $\alpha$ node) if all nodes in $\bar A \cup \{ \boldsymbol{u} \}$ were changed to type $\alpha$.   
\end{defin}   

\noindent So the point is that if a $\beta$ node belongs to an $\alpha$-intractable structure in the initial configuration then it can never change type. 

\subsubsection*{The proof of Theorem \ref{main2d} (e)} The proof goes through word for word the same as that of Theorem \ref{main2d} (a), if one replaces $\tau_{\alpha},\tau_{\beta}$ everywhere with $1-\tau_{\alpha}$ and $1-\tau_{\beta}$ and if one replaces `$\alpha$-stable' or `$\beta$-stable' with  `$\beta$-intractable' or `$\alpha$-intractable'. 

\subsubsection*{The proof of Theorem \ref{main2d} (d)} We define various events, which take the place of
 $\mathtt{pn}^{\alpha}_{\gamma,\tau_{\alpha}}(\boldsymbol{u})$, $\mathtt{uh}^\beta_{\tau_{\beta}}(\boldsymbol{u})$, $\mathtt{ruh}^{\alpha}_{\tau}(\boldsymbol{u})$ and $\mathtt{rju}^{\alpha}_{\tau}(\boldsymbol{u})$ in the proof of Theorem \ref{main2d} (b). 

\begin{itemize} 
\item   
Consider the initial configuration. If $\boldsymbol{u}$ is a $\beta$ node,  we let $\underline{\mathtt{pn}}^{\alpha}_{\gamma,\tau_{\alpha}}(\boldsymbol{u})$ be the event that there exists some $\gamma$-partial  neighbourhood of $\boldsymbol{u}$, $A$ say, such that $\boldsymbol{u}$ would be unhappy (as a node of type $\alpha$) if all nodes in $\bar A \cup \{ \boldsymbol{u} \}$ were changed to type $\alpha$. Note that, roughly, this corresponds to $A$ containing at least $(1-\tau_{\alpha})(2w+1)^2$ many nodes of type $\beta$.  We also let $\underline{\mathtt{pn}}^{\alpha}_{\gamma,\tau_{\alpha}}(\boldsymbol{u})[s]$ be the corresponding event for the end of stage $s$, rather than the initial configuration.

\item  We say that $\mathtt{h}^{\beta}_{\tau_{\beta}}(\boldsymbol{u})$ holds if $\boldsymbol{u}$ would be happy in the initial configuration if its type was changed to $\beta$ (or remains $\beta$ if already of this type). Note that roughly this corresponds to $\mathcal{N}(\boldsymbol{u})$ containing at most $(1-\tau_{\beta})$ many nodes of type $\alpha$. We say that 
$\underline{\mathtt{ruh}}^{\alpha}_{\tau}(\boldsymbol{u})$ holds if $\mathcal{N}^{\diamond}(\boldsymbol{u})$ contains strictly less than $(1-\tau)(2w+1)(3w+1)$ many $\beta$ nodes.  We say that $\underline{\mathtt{rju}}^{\alpha}_{\tau}(\boldsymbol{u})$ holds, if $\underline{\mathtt{ruh}}^{\alpha}_{\tau}(\boldsymbol{u})$ holds, but this would no longer be true if any node of type $\alpha$ in $\mathcal{N}^{\diamond}(\boldsymbol{u})$ changed type.  
\end{itemize}

In what follows we use the notation $f(w)\simeq g(w)$ introduced in the proof of Lemma \ref{zeta}. The crucial observation is that $ \textbf{P} (\underline{\mathtt{pn}}^{\alpha}_{\gamma,\tau_{\alpha}}(\boldsymbol{u})) \simeq  \textbf{P} (\mathtt{pn}^{\alpha}_{\gamma,1-\tau_{\alpha}}(\boldsymbol{u}))$, 
and $\textbf{P}( \mathtt{h}^{\beta}_{\tau_{\beta}}(\boldsymbol{u})) \simeq \textbf{P}( \mathtt{uh}^{\beta}_{1-\tau_{\beta}}(\boldsymbol{u})) $, while $\textbf{P} (\underline{\mathtt{ruh}}^{\alpha}_{\tau}(\boldsymbol{u})) = \textbf{P} ( \mathtt{ruh}^{\alpha}_{1-\tau}(\boldsymbol{u}))$. 
We therefore have the following analogue of Lemma \ref{zeta}: 

\begin{lem}  \label{zeta2} 
Suppose that $1-\kappa>\tau_{\alpha}>0.5$ and $\tau_{\beta} \vartriangleright \tau_{\alpha}$. Choose $\tau$ such that $\tau_{\beta}\vartriangleright \tau>\tau_{\alpha}$ and $1-\kappa >\tau$. Then there exists $\gamma>\frac{1}{2}$ and $\zeta>1$ such that for $0\ll w\ll n$, and for a node $\boldsymbol{u}$ selected uniformly at random, $\textbf{P}(\underline{\mathtt{ruh}}^\alpha_{\tau}(\boldsymbol{u}))>\zeta^w \cdot \textbf{P}(\underline{\mathtt{pn}}^{\alpha}_{\gamma,\tau_{\alpha}}(\boldsymbol{u}))$ and   $\textbf{P}(\underline{\mathtt{ruh}}^\alpha_{\tau}(\boldsymbol{u}))>\zeta^w \cdot  \textbf{P}(\mathtt{h}^{\beta}_{\tau_{\beta}}(\boldsymbol{u}))$.
\end{lem} 

With Lemma \ref{zeta2} in place, the remainder of the proof then goes through almost identically to that for Theorem \ref{main2d} (b), with (for the new values $\tau_{\alpha},\tau_{\beta}$ satisfying the conditions in the statement of Theorem \ref{main2d} (d), and for $\tau$ chosen as in Lemma \ref{zeta2}),  $\underline{\mathtt{pn}}^{\alpha}_{\gamma,\tau_{\alpha}}(\boldsymbol{u})$, $\mathtt{h}^\beta_{\tau_{\beta}}(\boldsymbol{u})$, $\underline{\mathtt{ruh}}^{\alpha}_{\tau}(\boldsymbol{u})$ and $\underline{\mathtt{rju}}^{\alpha}_{\tau}(\boldsymbol{u})$ taking the place of  $\mathtt{pn}^{\alpha}_{\gamma,1-\tau_{\alpha}}(\boldsymbol{u})$, $\mathtt{uh}^\beta_{1-\tau_{\beta}}(\boldsymbol{u})$, $\mathtt{ruh}^{\alpha}_{1-\tau}(\boldsymbol{u})$ and $\mathtt{rju}^{\alpha}_{1-\tau}(\boldsymbol{u})$ respectively. Now, however, we look to show that $\boldsymbol{u}_0$ which is chosen uniformly at random, will very probably eventually be of type $\alpha$. 

So we choose $r_{\ast}$ satisfying $(\dagger_a)$ and $(\dagger_b)$ of Section \ref{overview}  and observe, in a manner entirely analogous to what took place before, that if $r\geq r_{\ast}$ then for $0\ll w \ll n$, any $\beta$ node $\boldsymbol{v}$ on the outer boundary of a $\alpha$-firewall of radius $rw$ at stage $s$ will be hopeful,  so long as  $\underline{\mathtt{pn}}^{\alpha}_{\gamma,\tau_{\alpha}}(\boldsymbol{v})[s]$ does not hold. Then we define $\mathcal{Q}_0$ and $\mathcal{Q}$ as before, but using $\pi(w)$ which is now defined in terms of 
$\textbf{P}(\underline{\mathtt{ruh}}^{\alpha}_{\tau}(\boldsymbol{u}))$.  Our new version of 
$\mathtt{Typical \ vicinity}$ holds iff: 
\begin{enumerate}[\hspace{0.5cm}1.] 
\item $\underline{\mathtt{ruh}}^{\alpha}_{\tau}(\boldsymbol{u}_0)$ does not hold; 
\item There exists $\boldsymbol{u}\in \mathcal{Q}_0$ such that $\underline{\mathtt{ruh}}^{\alpha}_{\tau}(\boldsymbol{u})$ holds; 
\item There does not exist any $\boldsymbol{u}$ in $\mathcal{Q}$ such that either of  $\underline{\mathtt{pn}}_{\gamma,\tau_{\alpha}}^{\alpha}(\boldsymbol{u})$ or $\mathtt{h}^{\beta}_{\tau_{\beta}}(\boldsymbol{u})$ hold. 
\end{enumerate} 

As before, we may conclude that  if $0\ll w \ll n$ then (our new version of)  $\mathtt{Typical\  vicinity}$ fails to hold with probability $\ll \epsilon_0$.  In our new version of Definition \ref{smoothdef}, we replace ``type $\alpha$'' with ``type $\beta$''. In Definition \ref{counters} we again replace ``type $\alpha$'' with ``type $\beta$'' and vice versa, and we replace $\tau$ with $1-\tau$, so that these functions now detail the proportion of nodes of type $\beta$ (rather than $\alpha$ as previously was the case). In Definition \ref{k1} and Lemma \ref{firewall} we replace $\tau$ and $\tau_{\alpha}$ with $1-\tau$ and $1-\tau_{\alpha}$ respectively and we exchange ``type $\alpha$'' everywhere for ``type $\beta$'', and vice versa. Then the proof of our new version of Lemma \ref{firewall} goes through as before, except that now we are inductively able to argue that all nodes in the relevant sets $R^i_j$ are eventually of type $\alpha$, and that $C^0(\boldsymbol{u}) $ will eventually be a firewall of type $\alpha$.  Finally, using the same argument as in the proof of Lemma \ref{final} but using our new values of $\mathcal{Q}_0$, $m_0$ and so on, and replacing $\beta$ with $\alpha$, we conclude that if $\mathtt{Typical  \ vicinity}$ and $\mathtt{Smooth}(k_1,\epsilon_1)$ both hold then there exists a stage after which $\boldsymbol{u}_0$ is always of type $\alpha$. 

\section{Modifying the proofs for the three dimensional model} \label{sec five}

First note that Theorems  \ref{main3d} (a)  and \ref{main3d} (c) can be proved exactly as in the two dimensional case, if we simply modify all of the two dimensional notions to their three dimensional counterparts in the obvious way. Circles are replaced by spheres and rectangles by three dimensional blocks. We also extend the notation $A^{\dagger}$ to three dimensional space in the obvious way. Then Theorems   \ref{main3d} (d)  and \ref{main3d} (e) will once again follow by the same symmetry considerations  that we applied previously in Section \ref{upper half}, once Theorem \ref{main3d} (b) is established. Once again then, the bulk of the work is in establishing Theorem \ref{main3d} (b).

To establish Theorem \ref{main3d} (b), we can apply \emph{almost} exactly the same proof as in the two dimensional case. In that previous case, however, we worked quite hard in order to give a low value for $\kappa$, resulting in a slightly fiddly argument for the proof of Lemma \ref{firewall}, which does not obviously extend to three dimensional space. For the three dimensional model,  we shall be a little lazy and shall work with an analogue of the event   $\mathtt{ruh}^{\alpha}_{\tau}(\boldsymbol{u})$ which gives a value for $\kappa_{\ast}$ which could easily be improved upon, but which allows for a very simple proof in our three dimensional counterpart to Lemma \ref{firewall}.  

We redefine $\mathtt{pn}^{\alpha}_{\gamma,\tau}(\boldsymbol{u})$ and 
$\mathtt{uh}^{\alpha}_{\tau}(\boldsymbol{u})$ in the obvious way, and also define our analogue of  $\mathtt{ruh}^{\alpha}_{\tau}(\boldsymbol{u})$ as follows: 

\begin{itemize} 
\item Suppose that $P$ is a plane passing through $\mathcal{N}^{\dagger}(\boldsymbol{u})$. Let $A_1$ be all those points in $\mathcal{N}^{\dagger}(\boldsymbol{u})$ on or above $P$, and let $A_2$ be all those points in $\mathcal{N}^{\dagger}(\boldsymbol{u})$ on or below $P$. If $\mu(A_i^{\dagger})=\gamma(2w+1)^3$, then we call $A_i$ a $\gamma$-partial neighbourhood of $\boldsymbol{u}$, with defining plane $P$.   
If $\boldsymbol{u}$ is an $\alpha$ node,  we let $\mathtt{pn}^{\alpha}_{\gamma,\tau}(\boldsymbol{u})$ be the event that there exists some $\gamma$-partial  neighbourhood of $\boldsymbol{u}$ which  contains at least $\tau(2w+1)^3$ many $\alpha$ nodes in the initial configuration (and similarly for $\beta$).

\item  We say that $\mathtt{uh}^{\alpha}_{\tau}(\boldsymbol{u})$ holds if there are strictly less than $\tau(2w+1)^3$ many $\alpha$ nodes in $\mathcal{N}(\boldsymbol{u})$ in the initial configuration.   If  $\boldsymbol{u}=(x,y,z)$ then the \emph{extended neighbourhood} of $\boldsymbol{u}$  is the set of nodes $(x',y',z')$ such that $|x-x'|, |y-y'||z-z'|\leq \lceil \frac{3}{2}w \rceil$.  We say that  $\mathtt{euh}^{\alpha}_{\tau}(\boldsymbol{u})$  holds if there are strictly less than $\tau(3w+1)^3$ many $\alpha$ nodes in the extended neighbourhood of $\boldsymbol{u}$ in the initial configuration.
\end{itemize}

Arguing almost exactly as in the proof of Lemma \ref{zeta}, we then obtain the following three dimensional analogue: 

\begin{lem}  \label{zetathree} 
Suppose that $\kappa_{\ast}<\tau_{\alpha}<0.5$ and $\tau_{\beta} \leftslice \tau_{\alpha}$. Choose $\tau$ such that $\tau_{\beta}\leftslice \tau<\tau_{\alpha}$ and $\kappa_{\ast} <\tau$. Then there exists $\gamma>\frac{1}{2}$ and $\zeta>1$ such that for $0\ll w\ll n$, and for a node $\boldsymbol{u}$ selected uniformly at random, $\textbf{P}(\mathtt{euh}^\alpha_{\tau}(\boldsymbol{u}))>\zeta^w \cdot \textbf{P}(\mathtt{pn}^{\alpha}_{\gamma,\tau_{\alpha}}(\boldsymbol{u}))$ and   $\textbf{P}(\mathtt{euh}^\alpha_{\tau}(\boldsymbol{u}))>\zeta^w \cdot  \textbf{P}(\mathtt{uh}^{\beta}_{\tau_{\beta}}(\boldsymbol{u}))$.
\end{lem} 

The only significant task in modifying the proof to work in three dimensions is then in forming the analogue of Lemma \ref{firewall}, which is now much easier. We define $\mathtt{eju}^{\alpha}_{\tau}(\boldsymbol{u})$ in terms of $\mathtt{euh}^{\alpha}_{\tau}(\boldsymbol{u})$, just as $\mathtt{rju}^{\alpha}_{\tau}(\boldsymbol{u})$ was defined in terms of $\mathtt{ruh}^{\alpha}_{\tau}(\boldsymbol{u})$, i.e.\  $\mathtt{eju}^{\alpha}_{\tau'}(\boldsymbol{u})$ holds if there are less than $\tau'(3w+1)^3$ many $\alpha$ nodes in the extended neighbourhood of $\boldsymbol{u}$, but  changing the type of $(3w+1)^2$ $\beta$ nodes in this extended neighbourhood would cause this not to be the case. Rather than inductively building a sequence of larger and larger rectangles for which all nodes will become of type $\beta$, we simply observe that satisfaction of $\mathtt{eju}^{\alpha}_{\tau}(\boldsymbol{u})$ and (the three dimensional analogue of) $\mathtt{Smooth}_{k_1,\epsilon_1}(\boldsymbol{u})$ for appropriately large $k_1$ and small $\epsilon_1$, suffices to ensure that all $\alpha$ nodes in a cube centred at $\boldsymbol{u}$ with sides of length $w$, are unhappy in the initial configuration. In particular this certainly implies that all nodes in the interior of a sphere of diameter $w$ centred at $\boldsymbol{u}$ will be of type $\beta$ by the end of stage 1. Now satisfaction of $\mathtt{Smooth}_{k_1,\epsilon_1}(\boldsymbol{u})$ suffices to ensure that all $\alpha$ nodes on the outer boundary of this sphere are hopeful at the end of stage 1, and it then follows that this sphere of $\beta$ nodes will grow at subsequent stages as required.

\section{Deferred proofs} \label{defer} 

\subsection{{\bf The proof of Lemma \ref{zeta}}} We previously defined $\gamma$-partial neighbourhoods and right-extended neighbourhoods. It will be useful to consider also some related notions:

\begin{itemize} 
\item By the \emph{lower neighbourhood} of a node $\boldsymbol{u}=(x,y)$, denoted $\mathtt{ln}(\boldsymbol{u})$, we mean  the set of nodes $\{ (x',y')\in \mathcal{N}(\boldsymbol{u}):\ y'<y\ \mbox{or}\ (y'=y\ \& \ x'\leq x)\}$.
If $\boldsymbol{u}$ is an $\alpha$ node and $\tau\in [0,1]$, we let $\mathtt{ln}^\alpha_{\tau}(\boldsymbol{u})$ be the event that $\mathtt{ln}(\boldsymbol{u})$  contains at least $\tau(2w+1)^2$ many $\alpha$ nodes in the initial configuration. So $\mathtt{ln}^\alpha_{\tau_{\alpha}}(\boldsymbol{u})$ is the event that the lower neighbourhood already has enough $\alpha$ nodes to ensure that $\boldsymbol{u}$ is happy.  We let 
$\mathtt{ln}^{\beta}_{\tau}(\boldsymbol{u})$ denote the corresponding event when  $\alpha$ is replaced everywhere by $\beta$.\footnote{One should think of $\mathtt{ln}$ as (standing for) l-ower n-eighbourhood, and $\mathtt{rn}$ as r-otated n-eighbourhood.}   

\item Suppose that $\ell$ is any straight line passing through $\boldsymbol{u}$. Let $A_1$ be all those nodes in $\mathcal{N}(\boldsymbol{u})$ strictly below $\ell$, and let  $A_2$ be all those nodes in $\mathcal{N}(\boldsymbol{u})$ strictly above $\ell$. Suppose that $\ell$ intersects the boundary of $\mathcal{N}^{\dagger}(\boldsymbol{u})$ at $\boldsymbol{x}_1$ and $\boldsymbol{x}_2$ (ordered arbitrarily). Let $\ell_i$ be the (closed) straight line segment between $\boldsymbol{x}_i$ and $\boldsymbol{u}$. If $L \subset \mathcal{N}(\boldsymbol{u})$ is such that there exist $i,j \in \{ 0,1 \}$, $L= A_i \cup (\ell_j \cap \mathcal{N}(\boldsymbol{u}))$ then we call $L$ a \emph{rotated lower neighbourhood} of $\boldsymbol{u}$, with defining line $\ell$. If $\boldsymbol{u}$ is an $\alpha$ node and $\tau\in [0,1]$,  we let $\mathtt{rn}^{\alpha}_{\tau}(\boldsymbol{u})$ be the event that there exists some rotated lower neighbourhood of $\boldsymbol{u}$ which  contains at least $\tau(2w+1)^2$ many $\alpha$ nodes in the initial configuration (and similarly for $\beta$). 
\end{itemize}

\begin{figure}
 \centering
 \scalebox{0.4}{
\begin{tikzpicture}
  \filldraw[fill=green!20,draw=green!50!black] (0,0) -- (8,0) -- (8,4) -- (0,4) -- (0,0) ;
 \draw[very thin,color=gray] (0,0) grid (8,8);
\node[fill=red,circle] (b) at (4,4) {};
 \draw[black,very thick]    (0,4) -- (4,4);
 \draw[black,very thick, dashed]    (4,4) -- (8,4);
\end{tikzpicture}\hspace{1.5cm}
\begin{tikzpicture}
  \filldraw[fill=green!20,draw=green!50!black] (1.7,8) -- (6.3,0) -- (0,0) -- (0,8) -- (1.7,8);
 \draw[very thin,color=gray] (0,0) grid (8,8);
\node[fill=red,circle] (b) at (4,4) {};
 \draw[black,very thick]    (1.7,8) -- (4,4);
 \draw[black, dashed, very thick]    (4,4) -- (6.3,0);
\end{tikzpicture}\hspace{1.5cm}
\begin{tikzpicture}
  \filldraw[fill=green!20,draw=green!50!black] (0,0) -- (0,5.3) --  (8,4.7) -- (8,0) -- (0,0) ;
 \draw[very thin,color=gray] (0,0) grid (8,8);
\node[fill=red,circle] (b) at (4,4) {};
 \draw[black,very thick]    (0,5.3) -- (8,4.7);
 \end{tikzpicture}}

\caption{Lower, rotated lower and $\gamma$-partial neighbourhoods respectively. }
\label{fig:neighbourhoods}
\end{figure}
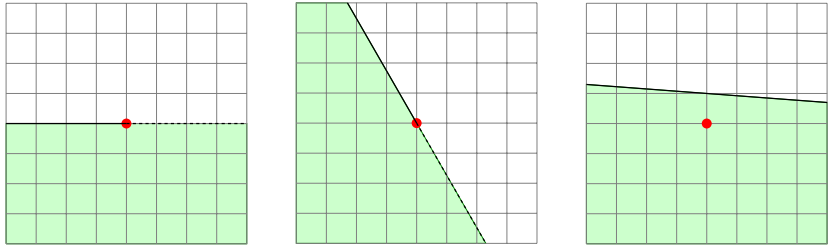

So the various neighbourhoods with which we are concerned are really very simple, and are depicted in Figure \ref{fig:neighbourhoods}. The lower neighbourhood is (roughly) just the bottom half of the neighbourhood, while the rotated lower neighbourhood is just a rotation of this. Recall that a $\gamma$-partial neighbourhood of $\boldsymbol{u}$ is just a subset of $\mathcal{N}(\boldsymbol{u})$ defined by a straight line, and which has measure $\gamma(2w+1)^2$. The right extended neighbourhood results roughly from tacking on an extra half of $\mathcal{N}(\boldsymbol{u})$ to the right.

We prove the following for $\boldsymbol{u}$ chosen uniformly at random: 

\begin{enumerate} 
\item $\tau >\kappa$ implies that there exists $\zeta>1$ such that for $0\ll w \ll n$,  $\textbf{P}(\mathtt{ruh}^\alpha_{\tau}(\boldsymbol{u}))>\zeta^w \cdot \textbf{P}(\mathtt{ln}^{\alpha}_{\tau}(\boldsymbol{u}))$. 
\item $\tau_{\beta}\lhd \tau$ implies that  there exists $\zeta>1$ such that for $0\ll w \ll n$, $\textbf{P}(\mathtt{ruh}^\alpha_{\tau}(\boldsymbol{u}))>\zeta^w \cdot  \textbf{P}(\mathtt{uh}^{\beta}_{\tau_{\beta}}(\boldsymbol{u}))$.
\item $\tau_\alpha >\tau$ implies that there exists $\gamma>\frac{1}{2}$ such that for $0\ll w \ll n$, $\textbf{P}(\mathtt{ln}^{\alpha}_{\tau}(\boldsymbol{u}))>\textbf{P}(\mathtt{pn}^{\alpha}_{\gamma, \tau_\alpha}(\boldsymbol{u}))$. 
\end{enumerate} 
\noindent
With these facts established,  (1) and (3) combined then suffice to show that, for the $\zeta$ guaranteed by (1), there exists $\gamma>\frac{1}{2}$ such that for $0\ll w \ll n$, $\textbf{P}(\mathtt{ruh}^\alpha_{\tau}(\boldsymbol{u}))>\zeta^w \cdot \textbf{P}(\mathtt{pn}^{\alpha}_{\gamma,\tau_{\alpha}}(\boldsymbol{u}))$.

First of all we note some basic facts about the binomial distribution. Define $b(N,k)= 2^{-N}{N \choose k}$, $b(N,\leq k)= \sum_{k'\leq k} b(N,k')$ and  $b(N,\geq k)= \sum_{k'\geq k} b(N,k')$. For the duration of this proof, we shall write $f(N)\simeq g(N)$ in order to indicate that there exist polynomials $P$ and $Q$ such that  $f(N)\cdot P(N)>g(N)$ and $g(N)\cdot Q(N)>f(N)$  for all $N$. We also write $f(N)\succsim g(N)$ to indicate that there exists $f'$ such that $f(N)\geq f'(N)$ for all $N$ and $f'(N)\simeq g(N)$. 
Now it follows, by direct inspection and from standard bounds on the tail of the binomial distribution (see for example Theorem 1.1. of \cite{BB}) that the following hold for fixed $0<\tau_1<0.5 <\tau_2 <1$ and for all integers  $c,d\in [-2,2]$: 

\begin{enumerate}[] 
\item $(\fgeU_1)$\hspace{0.5cm} $b(N, \leq \lfloor \tau_1 N \rfloor) \simeq b(N+c, \leq  \lfloor \tau_1 N \rfloor +d) 
\simeq b(N, \lfloor \tau_1 N \rfloor) \simeq b(N +c, \lfloor \tau_1 N \rfloor +d)$; 
\item $(\fgeU_2)$\hspace{0.5cm}  $b(N, \geq \lfloor \tau_2 N \rfloor) \simeq b(N+c, \geq  \lfloor \tau_2 N \rfloor +d) 
\simeq b(N, \lfloor \tau_2 N \rfloor) \simeq b(N +c, \lfloor \tau_2 N \rfloor +d)$. 
\end{enumerate} 

Now put $N=2w(w+1)$, so that $(2w+1)^2=2N+1$. For $\tau_{\alpha},\tau$ and $\tau_{\beta}$ in the range given by the conditions of the lemma, it follows from $(\fgeU_1)$ and $(\fgeU_2)$ that: 

\begin{enumerate}[] 
\item $(\fgeU_3)$\hspace{0.5cm} $\textbf{P}(\mathtt{ruh}^{\alpha}_{\tau}(\boldsymbol{u}))\succsim b(3N, \lceil 3\tau N \rceil)$, $\textbf{P}(\mathtt{ln}^{\alpha}_{\tau}(\boldsymbol{u})) \simeq b(N, \lceil 2\tau N \rceil)$ and 
$\textbf{P}(\mathtt{uh}^{\beta}_{\tau_\beta}(\boldsymbol{u})) \simeq b(2N, \lfloor \tau_{\beta} 2N \rfloor)$. 
 \end{enumerate} 

\subsection*{Proving (1)} From $(\fgeU_3)$ we get that $\textbf{P}(\mathtt{ruh}^{\alpha}_{\tau}(\boldsymbol{u}))/ \textbf{P}(\mathtt{ln}^{\alpha}_{\tau}(\boldsymbol{u})) )\succsim   b(3N, \lceil 3\tau N \rceil)/ b(N, \lceil 2\tau N \rceil)$. Applying Stirling's approximation, we see that the powers of $e$ cancel and: 

\[ \textbf{P}(\mathtt{ruh}^{\alpha}_{\tau}(\boldsymbol{u}))/ \textbf{P}(\mathtt{ln}^{\alpha}_{\tau}(\boldsymbol{u})) )\succsim  \frac{2^N (3N)^{3N+0.5} (2\tau N)^{2\tau N +0.5} (N(1-2\tau))^{N(1-2\tau)+0.5}}{2^{3N}(3N\tau)^{3N\tau +0.5}(3N(1-\tau))^{3N(1-\tau)+0.5}N^{N+0.5}}. \] 

\noindent Simplifying this we get: 

\begin{enumerate}[] 
\item $(\fgeU_4)$\hspace{0.5cm} $ \textbf{P}(\mathtt{ruh}^{\alpha}_{\tau}(\boldsymbol{u}))/ \textbf{P}(\mathtt{ln}^{\alpha}_{\tau}(\boldsymbol{u})) \succsim  \left( \frac{(1-2\tau)^{1-2\tau +1/(2N)}}{2^{2-2\tau-1/(2N)}\tau^{\tau}(1-\tau)^{3(1-\tau)+1/(2N)}} \right)^N. $
\end{enumerate} 

\noindent Now since $\tau>\kappa$ we have $(1-2\tau)^{1-2\tau} > 2^{2(1-\tau)}\tau^\tau (1-\tau)^{3(1-\tau)}$. Thus there exists $\zeta >1$ so that for all sufficiently large $N$ the term inside the brackets in $(\fgeU_4)$ is $>\zeta $, giving the result. 

\subsection*{Proving (2)}  From $(\fgeU_3)$ we get that $\textbf{P}(\mathtt{ruh}^{\alpha}_{\tau}(\boldsymbol{u}))/ \textbf{P}(\mathtt{uh}^{\beta}_{\tau_\beta}(\boldsymbol{u})) \succsim  b(3N, \lceil 3\tau N \rceil)/ b(2N, \lfloor \tau_\beta 2N \rfloor)$. Applying Stirling's approximation, and simplifying as before we get:

\begin{enumerate}[] 
\item $(\fgeU_5)$\hspace{0.5cm} $ \textbf{P}(\mathtt{ruh}^{\alpha}_{\tau}(\boldsymbol{u}))/ \textbf{P}(\mathtt{uh}^{\beta}_{\tau_\beta}(\boldsymbol{u})) \succsim \left( \frac{ \tau_{\beta}^{2\tau_\beta +1/(2N)} (1-\tau_{\beta})^{2(1-\tau_{\beta})+1/(2N)}}{2^{1-1/(2N)}\tau^{3\tau +1/(2N)}(1-\tau)^{3(1-\tau)+1/(2N)}} \right)^N. $
\end{enumerate} 

\noindent Now since $\tau_{\beta} \lhd \tau$ we have that $ \tau_{\beta}^{2\tau_\beta} (1-\tau_{\beta})^{2(1-\tau_{\beta})} > 2 \tau^{3\tau}(1-\tau)^{3(1-\tau)}$. Thus there exists $\zeta >1$ so that for all sufficiently large $N$ the term inside the brackets in $(\fgeU_5)$ is $>\zeta $, giving the result.

\subsection*{Proving (3)} Choose $\tau'$ with $\tau<\tau'<\tau_{\alpha}$. Note first that there exists $\zeta>1$ such that for $0\ll w\ll n$, $\textbf{P}(\mathtt{ln}^{\alpha}_{\tau}(\boldsymbol{u}))>\zeta^w \cdot \textbf{P}(\mathtt{ln}^{\alpha}_{\tau'}(\boldsymbol{u}))$. We then look to show that $\textbf{P}(\mathtt{rn}^{\alpha}_{\tau'}(\boldsymbol{u}))$ is at most $2(2w+1)^2 \cdot \textbf{P}(\mathtt{ln}^{\alpha}_{\tau'}(\boldsymbol{u}))$ -- meaning that once again there exists  $\zeta>1$ such that for $0\ll w\ll n$, $\textbf{P}(\mathtt{ln}^{\alpha}_{\tau}(\boldsymbol{u}))>\zeta^w \cdot \textbf{P}(\mathtt{rn}^{\alpha}_{\tau'}(\boldsymbol{u}))$. In order to see this, observe first that each rotated lower neighbourhood of $\boldsymbol{u}$ contains precisely the same number of nodes, and so is precisely as likely to contain $\tau'(2w+1)^2$ many nodes of type $\alpha$. As the defining line $\ell$  rotates through an angle of $2\pi$, each node leaves the corresponding rotated lower neighbourhood precisely once and also joins precisely once. Thus there are at most $2(2w+1)^2$ distinct rotated lower neighbourhoods of $\boldsymbol{u}$ and the probability that at least one of them contains at least $\tau'(2w+1)^2$ many nodes of type $\alpha$ is at most $2(2w+1)^2$ times the probability that a given one does. Finally we show that, for an appropriate choice of $\gamma>\frac{1}{2}$, there exists a constant $c$ such that $\textbf{P}(\mathtt{pn}_{\gamma,\tau_\alpha}^{\alpha}(\boldsymbol{u}))<c\cdot \textbf{P}(\mathtt{rn}_{\tau'}^{\alpha}(\boldsymbol{u}))$. So define: 
\[ \gamma = \frac{\tau_{\alpha}+\tau'}{4\tau'}. \] 
Suppose we are given that $\mathtt{pn}^{\alpha}_{\gamma,\tau_{\alpha}}(\boldsymbol{u})$ holds. Given $A$ which is a $\gamma$-partial neighbourhood of $\boldsymbol{u}$ containing at least $\tau_{\alpha}(2w+1)^2$ many nodes of type $\alpha$, with defining line $\ell$ say, consider a rotated lower neighbourhood of $\boldsymbol{u}$  which is a subset of $A$, with defining line parallel to $\ell$. Call this rotated lower neighbourhood $A_0$.  Let $z_0$ be the expected number of $\alpha$ nodes in $A_0$ (given the stated conditions on $A$), and let $z_1$ be the actual number of $\alpha$ nodes in $A_0$. As $w\rightarrow \infty$: 

 \[ z_0/(2w+1)^2\rightarrow z^{\ast}:=\frac{2\tau'\tau_{\alpha}}{(\tau_{\alpha}+\tau')}, \] 
 and for any $\epsilon>0$, the probability that $|z_1/(2w+1)^2-z^{\ast}|>\epsilon$ tends to 0. The result follows since $(2\tau_{\alpha})/((\tau_{\alpha}+\tau'))>1$.

\subsection{{\bf The proof of Lemma \ref{smoothlem}}} If $\theta(w)$ is the proportion of the nodes in $\mathcal{N}^{\diamond}(\boldsymbol{u})$ which are of type $\alpha$ when $ \mathtt{rju}^{\alpha}_{\tau}(\boldsymbol{u})$ holds, then $\theta(w)\rightarrow \tau$ as $w\rightarrow \infty$. Since we consider $k_1$ and $\epsilon_1$ fixed it suffices to show that for any $\epsilon>0$ and for a given $A\in \mathcal{I}_{k_1}(\boldsymbol{u})$ the following occurs with probability $>1-\epsilon$ for $0\ll w \ll n$: 
 
 \begin{enumerate} 
\item If $A\subset \mathcal{N}^{\diamond}(\boldsymbol{u})$ then the proportion of the nodes in $A$ which are of type $\alpha$ is in the interval $[\tau -\epsilon_1, \tau+\epsilon_1]$. 
\item If $A$ is disjoint from  $\mathcal{N}^{\diamond}(\boldsymbol{u})$ then the  proportion of the nodes in $A$ which is of type $\alpha$ is in the interval $[0.5 -\epsilon_1, 0.5+\epsilon_1]$. 
\end{enumerate} 

Now (2) follows simply from the weak law of large numbers. Also,  (1) follows  from Chebyshev's inequality and standard results for the variance of a hypergeometric distribution. If $\phi$ is the proportion of the nodes in $A$ which are of type $\alpha$ then: 
\[ \textbf{P}(|\phi- \theta(w)|>\epsilon_1/2)<|A|^{-2}(\epsilon_1/2)^{-2} \mbox{Var}(\theta)= O(1)|A|^{-1}. \]


 \bibliographystyle{alpha}

\end{document}